\documentclass[12pt]{article}
\usepackage[left=2.5cm, right=2.5cm, top=2.50cm, bottom=2.5cm]{geometry}
\usepackage{authblk,amsmath,amsfonts,amssymb,amsthm, physics}
\usepackage{algorithm}
\usepackage{algpseudocode}
\usepackage{graphicx, subfigure, tikz, makecell, multirow}
\usepackage{comment}
\usepackage{epsfig}
\usepackage{adjustbox}
\usepackage{nicefrac}
\usepackage{mathrsfs}
\usepackage{float}
\usepackage{booktabs}
\usepackage{threeparttable}
\usepackage{tabularray}
\usepackage{enumitem}  
\usepackage[utf8]{inputenc}
\usepackage[english]{babel}
\usepackage{pdfpages}
\usepackage{bbm}
\usepackage{appendix}
\usepackage{esint}
\usepackage{mathtools}
\usepackage{colortbl, color}
\usepackage{hyperref}[colorlinks,linkcolor=blue]
\numberwithin{equation}{section}

\newcommand{\e}{\mathrm{e}}

\newcommand{\bbE}{\mathbb{E}}
\newcommand{\bbP}{\mathbb{P}}
\newcommand{\bbR}{\mathbb{R}}
\newcommand{\bbQ}{\mathbb{Q}}
\newcommand{\bbN}{\mathbb{N}}
\newcommand{\bbY}{\mathbb{Y}}
\newcommand{\bbZ}{\mathbb{Z}}
\newcommand{\de}{\mathrm{d}}
\newcommand{\Nexp}{N_{\text{exp}}}
\newcommand\fa[1]{\left[#1\right]}
\newcommand\br[1]{\left(#1\right)}

\newcommand\hua[1]{\left\{#1\right\}}

\newcommand\bm[1]{\boldsymbol{#1}}

\def\calB{\mathcal{B}}

\def\calF{\mathcal{F}}

\def\calL{\mathcal{L}}

\def\calN{\mathcal{N}}
\def\calO{\mathcal{O}}

\def\calP{\mathcal{P}}

\def\calV{\mathcal{V}}

\usepackage{xcolor}
\definecolor{myorange}{RGB}{201,53,56}
\hypersetup{colorlinks=True,citecolor=myorange}

\newtheorem{remark}{Remark}[section]
\newtheorem{theorem}{Theorem}[section]

\newtheorem{assumption}{Assumption}[section]
\newcommand*\xbar[1]{%
	\hbox{%
		\vbox{%
			\hrule height 0.5pt 
			\kern0.5ex
			\hbox{%
				\kern-0.1em
				\ensuremath{#1}%
				\kern-0.1em
			}%
		}%
	}%
}

\def\XXint#1#2#3{{\setbox0=\hbox{$#1{#2#3}{\int}$}
		\vcenter{\hbox{$#2#3$}}\kern-.5\wd0}}

\title{Unsupervised Learning-based Calibration Scheme for Rough Volatility Models}
\author{Changqing Teng \thanks{Department of Mathematics, The University of Hong Kong, Pokfulam Road, Hong Kong. Email: {\tt{u3553440@connect.hku.hk}}.} \quad and \quad Guanglian Li\thanks{Corresponding author. Department of Mathematics, The University of Hong Kong, Pokfulam Road, Hong Kong. Email: {\tt{lotusli@maths.hku.hk}}.}}
\begin{document}
	
	\maketitle	
	
	\begin{abstract}
		Existing deep learning-based calibration scheme for rough volatility models predominantly rely on supervised learning frameworks, which incur significant computational costs due to the necessity of generating massive synthetic training datasets. In this work, we propose a novel unsupervised learning-based calibration scheme for rough volatility models that eliminates the data generation bottleneck. Our approach leverages the backward stochastic differential equation (BSDE) representation of the pricing function derived by Bayer et al. \cite{bayer2022pricing}. By treating model parameters as trainable variables, we simultaneously approximate the BSDE solution and optimize the parameters within a unified neural network training process, with the terminal misfit as the loss. We theoretically establish that the mean squared error between the model-implied prices and market data is bounded by the loss function. Furthermore, we prove that the loss can be minimized to an arbitary degree, depending on the model's market fitting capacity and the universal approximation capability of neural networks. Numerical experiments for both simulated and historical S\&P 500 data based on rough Bergomi (rBergomi) model demonstrate the efficiency and accuracy of the proposed scheme.
	\end{abstract}
	

	\section{Introduction}
	Since the pioneering work \cite{gatheral2018volatility}, there has been significant and growing interest in studying rough volatility models \cite{abi2019multifactor, bayer2020regularity, bayer2019short, forde2017asymptotics,fukasawa2021volatility,fukasawa2019volatility,horvath2021deepa,livieri2018rough, romer2022empirical}, owing to their statistical consistency and effectiveness in option pricing. Consider a filtered probability space $(\Omega, \calF, \br{\calF_t}_{t \in [0, T]}, \bbQ)$ with $\bbQ$ being the risk-neutral measure. The general dynamics of a rough volatility model are given by
	\begin{align}\label{eq: underlying}
		\de S_t = rS_t\de t + S_t\sqrt{V_t}\br{\rho \de W_t + \sqrt{1 - \rho^2}\de W_t^\perp},\quad S_0 = s_0,		
	\end{align}   
	where $t \in [0, T]$ and $T \in (0, +\infty)$ represents the terminal time. The constant $r$ denotes the interest rate, while $W_t$ and $W_t^\perp$ are independent standard Brownian motions. We denote by $\br{\calF_t}_{t \in [0, T]}$ the augmented filtration generated by $W$ and $W^\perp$, and by $\br{\calF_t^W}_{t \in [0, T]}$ the augmented filtration generated by $W$ alone.
	The variance process $V$ satisfies the following assumption:
	\begin{assumption}\label{assump: variance process}	
		The variance process $V$ possesses a.s. continuous, non-negative trajectories and is adapted to $\mathcal{F}_t^W$. Furthermore, $V_t$ is integrable; that is, 
		\begin{align}\label{eq: regularity of V 1}
			\bbE\fa{\int_0^T V_s \de s} < \infty, \quad T>0. 
		\end{align}	
	\end{assumption}	
	A representative example of the rough volatility paradigm satisfying Assumption \ref{assump: variance process} is the \textit{rough Bergomi (rBergomi) model} \cite{bayer2016pricing}, where the variance is defined as
	\begin{align}\label{eq: rBergomi variance}
		V_t = \xi_0(t)\exp\br{\eta\sqrt{2H} \int_{0}^{t}(t-s)^{H - \frac{1}{2}}\de W_s- \frac{\eta^2}{2}t^{2H}},\quad V_0 = v_0,
	\end{align}
	where $\xi_0(t) := \mathbb{E}[V_t|\calF_0]$ represents the initial forward variance curve. The term "rough" derives from the Hurst index $H \in (0, {1}/{2})$, as the sample paths of $V_t$ are a.s. $\br{H - \varepsilon}$-H\"older continuous for any $\varepsilon > 0$, exhibiting lower regularity than standard Brownian motion. The correlation $\rho \in (-1, 0)$ is ensures the discounted price $\e^{-rt}S_t$ remains a martingale \cite{gassiat2019martingale} and $\eta > 0$ denotes the volatility of volatility. Another prominent model is the \textit{rough Heston model} \cite{el2019characteristic}, a rough extension of the celebrated Heston model, governed by the following variance dynamics:
	\begin{align*}
		V_t = V_0 + \frac{1}{\Gamma(H+1/2)}\int_0^t (t-s)^{H-1/2}(\theta - \lambda V_s)\de s + \frac{1}{\Gamma(H+1/2)}\int_0^t (t-s)^{H-1/2}\nu\sqrt{V_s}\de W_s,
	\end{align*}
	This model also satisfies Assumption \ref{assump: variance process}, as shown in \cite{abi2019affine}.
	
	The practical utility of any financial model relies heavily on the availability of efficient calibration methods. Calibration entails identifying the model parameters $\theta$ within a set $\Theta \in \bbR^n$ (for some $n \in \bbN$) such that the model-generated option prices align closely with market data. Specifically, let $\zeta = (K, T) \in \mathbb{R}_+^2$ denote the strike and maturity representing the contract details. We define the pricing map $\mathcal{P}: \Theta \times \mathbb{R}_+^2 \to \mathbb{R}_+, (\theta, \zeta) \mapsto \calP(\theta, \zeta)$ to demonstrate that the option price is fully determined by the model parameters and the contract specifications. The calibration objective is to solve the following optimization problem:
	\begin{align} \label{eq:objective1} 
		\inf_{\theta \in \Theta} \frac{1}{M}\sum_{m=1}^M\left(\mathcal{P}(\theta, \zeta_m) - P_{MKT}(\zeta_m)\right)^2, 
	\end{align}
	where $M$ represents the number of market data points. Addressing the calibration task typically requires multiple iterations of the model, underscoring the necessity for an efficient pricing methodology. However, the inclusion of the Riemann-Liouville type of fractional Brownian motion, given by
	\begin{align*}
		W^H_t := \sqrt{2H}\int_0^t (t-s)^{H - \frac{1}{2}}\mathrm{d} W_s, 
	\end{align*}
	renders the joint process $(S_t, V_t)$ non-Markov and deprives it of a semi-martingale structure. This characteristic precludes the application of standard PDE-based methods and all techniques relying on the Feynman-Kac representation theorem. While the characteristic 
	function of the rough Heston model can be expressed via a fractional Ricatti equation, enabling classical methods \cite{el2019characteristic}, the rBergomi model lacks this affine structure. Consequently, pricing methods for rBergomi predominantly rely on Monte Carlo simulations. Despite the recent advancements \cite{bennedsen2017hybrid, teng2025efficient}, substantial sampling is still required during the online calibration phase. 
	
	In recent years, data-driven concepts and deep learning approaches using neural networks have been extensively explored to accelerate calibration. Existing deep learning-based calibration schemes generally fall into two categories. The first is the \textit{one-step approach} or direct inverse mapping \cite{hernandez2016model}, where a neural network is trained to learn the map:
	\begin{align*}
		\Phi: \bbR^M_+ \to \Theta,\quad \br{P_{MKT}(\zeta_m)}_{m=1}^M \mapsto \theta,
	\end{align*}
	taking historical data as input and model parameters $\theta$ as the output. Traditional calibration is performed on historical data to generate labeled training pairs:
	\begin{align*}
		\br{\br{P_{MKT}^i\br{\zeta_m}}_{m=1}^M, \theta^i}_{i = 1}^{N_{train}}, 
	\end{align*}
	where the number of training pairs $N_{train}$ is intrinsically limited by the volume of reliable market data. Once trained, the network can directly output parameters for any given market data, offering maximum efficiency. While straightforward and capable of fitting observations well, this method lacks control and generalizes poorly to unseen data, as notedy by \cite{bayer2019deep, hernandez2016model}. The second category is the \textit{two-step approach} (learning the model, then calibrating to data) \cite{bayer2019deep, horvath2021deepb,liu2019neural, stone2020calibrating, rosenbaum2021deep}. Here, a neural network approximates the pricing function $\calP$, mapping model parameters to vanilla option prices (or implied volatilities) for a prespecified grid of maturities and strikes:
	\begin{align*}
		\Psi: \Theta \to \bbR_+^M,\quad \theta \mapsto ( \tilde{\calP}(\theta, \zeta_m))_{m =1}^{M}.
	\end{align*}	
	Training pairs are synthetically generated, bypassing the limitations of available market data. It has been established that the neural networks can learn the pricing function with an arbitrarily small error $\varepsilon$, while the network size grows only sub-polynomially in $1/\varepsilon$ under suitable assumptions \cite{biagini2024approximation}. With this deterministic approximation, the calibration problem \eqref{eq:objective1} is reformulated as
	\begin{align*}
		\inf_{\theta \in \Theta} \frac{1}{M}\sum_{m=1}^{M}\br{\tilde{\calP}(\theta, \zeta_m) - P_{MKT}(\zeta_m)}^2.
	\end{align*}	
	This approach shifts the computationally intensive numerical approximation to the offline stage, which facilitates faster online calibration processes.  
	
	Both approaches operates within a supervised learning framework, requiring large-scale labeled datasets. Despite promising results, they exhibit three key drawbacks. First, generating training data  relies heavily on Monte Carlo-based methods, which incurs significant computational costs and storage requirements. Second, the trained neural network requires market data to align with a prespecified strike and maturity grid. While it is possible to choose a fine grid, this increases complexity of training pairs generation \cite{horvath2021deepb}. Finally, the complexity and dynamism of financial markets make it difficult to capture all nuances with a single mathematical model. Parameters may follow unknown, application-specific distributions inferred from market data. For instance, the scalar parameter $H$ may vary over time to reflect changing local volatility regularity \cite{corlay2014multifractional}. Since existing approaches presuppose specific parameter forms for training, they may lack the flexibility to adapt to general cases or capture complete market information.
	
	The primary objective of this work is to propose an unsupervised learning-based calibration scheme to address these challenges. While vanilla option prices in classical diffusive stochastic volatility models solve linear parabolic partial differential equations (PDEs), allowing calibration to be viewed as a parameter identification (inverse problem), rough volatility models require a different approach. Due to their non-Markov nature, the corresponding option price follows a backward stochastic partial differential equation (BSPDE). A stochastic Feynman-Kac formula has been established to represent the weak solution of the BSPDE via a BSDE coupled with the forward SDE \cite{bayer2022pricing}, and deep learning schemes have been developed for backward option pricing. Our proposed method, the Deep BSDE scheme, treats model parameters as tunable weights optimized alongside BSDE solutions during training. Inspired by \cite{han2018solving, han2017deep}, our scheme implements a forward Euler method that initiates from market data and matches the terminal condition to solve the BSDE.
	
	The contribution of this work is twofold. First, we propose an unsupervised learning-based calibration scheme for rough volatility models that circumvents the challenge of training pair generation. Unlike purely data-driven methods such as the one-step approach, our scheme incorporates PDE theory to reduce reliance on extensive market data while maintaining high accuracy. Furthermore, it adapts to increasing market data and accommodates more general parameter forms. During training, model parameters and numerical BSDE solutions are learned simultaneously. Second, we demonstrate in Section \ref{sec: convergence} that the discrepancy in option prices implied by the model parameters is bounded by the loss function. Given the market-fitting capability of rough volatility models and the universal approximation property of neural networks, the loss can be effectively controlled. This analysis elucidates the relationship between parameter estimation accuracy and option price discrepancies, providing a deeper understanding of the calibration process.
	
	\begin{table}[htp]
		\centering
		\begin{adjustbox}{max width=\textwidth}
			\begin{tabular}{c cccc}
				\toprule
				Methods & PDE usage & \makecell{Training data \\ usage} & \makecell{Strike and Maturity \\grid} & \makecell{Forms of model \\ parameters}\\
				\hline
				One-step approach & No & Yes & Fixed & Fixed\\
				Two-step approach & No & Yes & Fixed & Fixed \\
				Deep BSDE &  Yes & No & Flexible & Flexible\\
				\bottomrule
			\end{tabular}
		\end{adjustbox}
		\caption{Comparison of deep learning-based calibration schemes.}
		\label{tab: DL concept comparison}
	\end{table}

	The remainder of this paper is organized as follows. Section \ref{sec:bsde} reviews the BSPDE developed in \cite{bayer2022pricing} for rough volatility models. Section \ref{sec:ulc} presents the complete calibration scheme. Convergence analysis is provided in Section \ref{sec: convergence}. In Section \ref{sec: exp}, we present extensive numerical experiments based on the rBergomi model using both synthetic and historical data, benchmarking against the pure Monte Carlo method. The algorithm and all numerical examples are available online at \url{https://github.com/evergreen1002/Calibration-RoughVol-BSDE}. We conclude with remarks and future research directions in Section \ref{sec:conclusion}.
	
	\section{BSPDE model}\label{sec:bsde}
	This section reviews the European option pricing theory for rough volatility models formulated via a Backward Stochastic Partial Differential Equation (BSPDE), as established by \cite{bayer2022pricing}. We define the process $X_s^{t,x} := -rs + \ln S_s$ with the initial state $x :=-rt+ \ln S_t$ for $0\leq t\leq s\leq T$. Following \eqref{eq: underlying}, the dynamics of the log-price are given by:
	\begin{equation}\label{eq: log price}
		\de X_s^{t,x} = -\frac{1}{2}V_s\de s + \sqrt{V_s}\br{\rho \de W_s + \sqrt{1 - \rho^2}\de W_s^\perp},\quad X_t^{t,x} =x. 
	\end{equation}
	The price of a European option at time $t$ based on $\eqref{eq: log price}$ with payoff function $h(\cdot)$ is defined as:	
	\begin{equation}\label{eq:martingale-price}
		u_t(x) := \bbE\fa{\e^{-r(T-t)}h\br{\exp(X_T^{t, x} +rT)}|\calF_t},
	\end{equation}
	which is a random field for $(t,x) \in [0, T] \times \bbR$. For a European call option, the payoff function $h(\cdot) = \br{\cdot - K}^+$. Given the non-Markovian nature of the pair $\br{X_t, V_t}$, representing $u_t(x)$ via a conventional deterministic PDE is not feasible. However, \cite{bayer2022pricing} established that the pair of random fields $(u, \psi)$ constitutes a weak solution to the following BSPDE in the sense of \cite[Definition 2.1]{bayer2022pricing}:
	\begin{equation}\label{eq: BSPDE}
		\begin{aligned}
			-\de u_t(x) &= \left(\frac{V_t}{2}D^2u_t(x) + \rho\sqrt{V_t}\psi_t(x) - \frac{V_t}{2}Du_t(x) -ru_t(x)\right)\de t - \psi_t(x)\de W_t,\,\\
			u_T(x) &= G(\e^x), 
		\end{aligned}	
	\end{equation}
	where $G(\e^x) := h(\e^{x + rT})$. To derive the stochastic Feynman-Kac formula presented in \cite{bayer2022pricing} which connects the BSPDE \eqref{eq: BSPDE} to a BSDE, we impose the following condition on $G$.
	\begin{assumption} \label{assump: payoff}
		The function $G: \br{\Omega \times \bbR, \calF_T^W \otimes \calB(\bbR)} \to \br{\bbR, \calB(\bbR)}$ satisfies for some constant $L>0$
		\begin{align*}
			G(x) \leq L(1 + \abs{x}),\quad x \in \bbR.	
		\end{align*}
	\end{assumption}
	\begin{theorem}[Stochastic Feynman-Kac formula {\cite[Theorem 2.4]{bayer2022pricing}}]
		Suppose Assumptions \ref{assump: variance process} and \ref{assump: payoff} hold. Let $\br{u, \psi}$ be a weak solution of the BSPDE $\eqref{eq: BSPDE}$ such that there is a constant $C_0 \in (0, \infty)$ satisfying for each $t \in [0, T]$
		\begin{align*}
			\abs{u_t(x)} \leq C_0\br{1 + \e^x}\quad \text{for almost all} \br{\omega, x} \in \Omega \times \bbR. 
		\end{align*}
		Then the following holds a.s., 
		\begin{equation}\label{eq: stochastic Feynman Kac}
			\begin{aligned}
				u_s(X_s^{t,x}) &= Y_s^{t, x}, \\
				\sqrt{(1 - \rho^2)V_s}Du_s(X_s^{t, x}) &= Z_s^{t,x},\\
				\psi_s(X_s^{t,x}) + \rho\sqrt{V_s}Du_s(X_s^{t, x}) &= \tilde{Z}_s^{t, x},  
			\end{aligned}
		\end{equation}
		for $0\leq t \leq s \leq T$ and $x \in \bbR$, where $(Y_s^{t, x}, Z_s^{t,x}, \tilde{Z}_s^{t, x})$ is the unique solution of the following BSDE in the sense of \cite[Definition 2.1]{briand2003lp}
		\begin{equation}\label{eq: BSDE}
			\begin{aligned}
				-\de Y_s^{t, x} = -rY_s^{t,x}\de s  - \tilde{Z}_s^{t, x}\de W_s - Z_s^{t,x}\de W_s^\perp,\quad 
				Y_T^{t, x} = G(\e^{X_T^{t,x}}). 
			\end{aligned}    
		\end{equation}
	\end{theorem}
	
	Without loss of generality, we set the initial time $t \equiv 0$ and $x_0 := \ln s_0$. Consequently, the system comprising \eqref{eq:martingale-price} and \eqref{eq: BSDE} yields the following decoupled Forward-Backward SDE (FBSDE):
	\begin{equation}\label{eq: FBSDE}
		\begin{aligned}			
			\de X_s &= -\frac{1}{2}V_s\de s + \sqrt{V_s}\br{\rho \de W_s + \sqrt{1 - \rho^2}\de W_s^\perp},&& X_0 = x_0,\\
			-\de Y_s &= -rY_s\de s - \tilde{Z}_s\de W_s - Z_s\de W_s^\perp,&&
			Y_T = G(\e^{X_T}). 
		\end{aligned}
	\end{equation}
	Here and in the sequel, we drop the superscripts to simplify the notation. This stochastic Feynman-Kac formula implies the uniqueness of a weak solution to the BSPDE \eqref{eq: BSPDE} which together with the existence presented in \cite[Theorem 2.5]{bayer2022pricing}, ensures the problem is well-posed. By the first equation in $\eqref{eq: stochastic Feynman Kac}$, the solution $u$ of the BSPDE corresponds to a specific solution of the BSDE \eqref{eq: BSDE}. Therefore, the option pricing problem \eqref{eq:martingale-price} is equivalent to solving the associated BSDE. 
	\begin{remark}[Functional dependence of BSDE solutions]\label{rm:variables}
		In light of \eqref{eq: stochastic Feynman Kac} and the fact that $Y_{t}^{t, x} \in \mathcal{F}^W_t$ \cite[Theorem 2.2]{bayer2022pricing}, $Y_s$ can be interpreted as a functional dependent on $ (\br{V_u}_{u \in [0, s]}, X_s)$. A similar interpretation applies to $Z_s$ and $\tilde{Z}_s$. Alternative input configuration include $((W_u)_{u\in [0, s]}, X_s)$, $((W^H_u)_{u \in [0, s]}, X_s)$ or $((W_u)_{u \in [0, s]}, (W^H_u)_{u \in [0,s]}, X_s)$ as discussed in \cite[Remark 4.3]{bayer2022pricing}. \footnote{We thank Jinniao Qiu for the comment regarding possible inputs.}
	\end{remark}

	\section{Unsupervised Learning-based Calibration method}\label{sec:ulc}
	In this section, we present the proposed unsupervised learning-based calibration scheme for rough volatility models. Let $\theta$ denote the model parameters, and let $\zeta$ represent the set of relevant market parameters, defined as
	\begin{align*}
		\zeta := \br{K_{\ell}, T_j}_{\substack{\ell = 1, \cdots, L\\ j =1, \cdots, N}},
	\end{align*}     
	where each pair $(K{\ell},T_j)$ denotes the strike price and expiry of a European option. We assume the maturities ${T_j}_{j=1}^N$ are sorted in ascending order. Inspired by the methodologies in \cite{han2018solving, han2017deep}, we solve the BSDE in a forward manner using deep learning techniques.

	\subsection{Deep learning based scheme}	
	We begin by discretizing the temporal domain $[0,T]$ (where $T = T_N$) using an equidistant grid $\pi: 0 = t_0 < t_1 < \cdots < t_n = T$ with step size $h := T/n$ and $t_i := ih$. We assume that each maturity $T_j$ coincides with a grid point, such that $T_j = k(j)h$ for a strictly increasing function $k: \mathbb{N} \to \mathbb{N}$. Let $\Delta W_{t_{i}} := W_{t_{i+1}} - W_{t_i}$ and $\Delta W_{t_{i}}^\perp := W_{t_{i+1}}^\perp - W_{t_i}^\perp$ denote the Brownian motion increments for $i = 0, \cdots, n-1$. We explicitly formulate the dependence of $(X_t, V_t)$ on the model parameters $\theta$ and approximate $X_t(\theta)$ in \eqref{eq: FBSDE} using the Euler-Maruyama scheme:
	\begin{equation}\label{eq: euler forward SDE}
		\begin{aligned}	
			X_{t_{i+1}}^\pi(\theta) &= X_{t_{i}}^\pi(\theta) - \frac{1}{2}V_{t_i}^\pi(\theta) h + \rho\sqrt{V_{t_i}^\pi(\theta)}\Delta W_{t_i} + \sqrt{(1 - \rho^2)V_{t_i}^\pi(\theta)}\Delta W_{t_i}^\perp,\quad X_0^\pi &= x_0,
		\end{aligned}    	
	\end{equation}
	where $\br{V_{t_i}^\pi(\theta)}_{i = 1, \cdots, n-1}$ is obtained by the mSOE scheme proposed in \cite{teng2025efficient}.
	
	Next, we solve the BSDE $\eqref{eq: BSDE}$, where $\zeta$ characterizes the terminal condition. Let $\bm{Y}^j$ be the $\bbR^L$-valued stochastic process with expiry $T_j$, satisfying the following BSDE 
	\begin{align*}
		-\de \bm{Y}_t^j &= -r\bm{Y}_t^j\de t - \widetilde{\bm{Z}}_t^j\de W_t - \bm{Z}_t^j\de W_t^\perp,\\
		\bm{Y}_{T_j}^j(\theta) &= \fa{\br{\e^{X_{T_j}(\theta) +rT_j} - K_1}^+, \cdots, \br{\e^{X_{T_j}(\theta) +rT_j} - K_L}^+}^\top,
	\end{align*}
	with $\tilde{\bm{Z}}^j, \bm{Z}^j \in \bbR^L$ for $j = 1, \cdots, N$. We explicitly express the payoff function $G(\cdot)$ to highlight its dependence on the strike, focusing specifically on call options. For notational simplicity, we denote $X_i$ for $X_{t_i}$, $V_i$ for $V_{t_i}$, $\bm{Y}^j_i$ for  $\bm{Y}^j_{t_i}$, $\bm{Z}_{i}^j$ for $\bm{Z}_{t_i}^j$ and $\tilde{\bm{Z}}_i^j$ for $\tilde{\bm{Z}}_{t_i}^j$ in the sequel. We employ the Euler-Maruyama scheme to approximate $\bm{Y}_i^j$, setting the market option prices as the initial values:
	\begin{align*}
		\bm{Y}_0^{j, \pi} &= \fa{P_{MKT}\br{K_1, T_j}, \cdots, P_{MKT}\br{K_L, T_j}}^\top, \\
		\bm{Y}_{i+1}^{j, \pi}(\theta) &= \br{1 + rh}\bm{Y}_{i}^{j, \pi}(\theta) + \tilde{\bm{Z}}_{i}^{j, \pi}(\theta)\Delta W_{t_i} + \bm{Z}_i^{j, \pi}(\theta)\Delta W_{t_i}^\perp, 	
	\end{align*}
	for $i = 0, \cdots, k(j)-1$. 
	To aggregate the scheme for all $j = 1, \cdots, N$ into a matrix format, we set the initial value of the numerical BSDE solution as:
	\begin{align}\label{eq:initial-bsde}
		\bbY_0^\pi := [\bm{Y}_0^{1,\pi}, \cdots, \bm{Y}_0^{N,\pi}] = [(P_{MKT}(K_\ell, T_j))]_{L\times N}\in \bbR^{L\times N}.
	\end{align}
	Similar, we implement forward Euler scheme to get $\bbY_1^\pi, \cdots, \bbY_i^\pi$. When $i = k(1)$, i.e. $t_i = T_1$, options with maturity $T_1$ expire. Consequently, for $i = k(1)+1$, $\mathbb{Y}_i^\pi$ excludes these options, and its size is reduced to ${L\times(N-1)}$. To be precise, let $j\br{i}  := \min\hua{j: k(j) > i}$ such that $T_{j(i)-1}<t_i\leq T_{j(i)}$, we then have 
	\begin{align}\label{eq: Y matrix}
		\bbY_i^\pi(\theta) := [\boldsymbol{Y}^{j(i), \pi}_i(\theta), \cdots, \boldsymbol{Y}^{N, \pi}_i(\theta)] \in \bbR^{L \times (N - j(i) + 1)}.
	\end{align}
	Following Remark \ref{rm:variables}, we interpret ${\tilde{\bm{Z}}}_i^{j,\pi}(\theta)$ and $\bm{Z}_i^{j,\pi}(\theta)$ as $\mathbb{R}^L$-valued functions of $\br{V_0, V_1^\pi(\theta), \cdots, V_{i}^\pi(\theta), X_{i}^\pi(\theta)}$ and employ neural networks as function approximators. In light of \eqref{eq: Y matrix}, we define
	\begin{equation}\label{eq: NN approx}
		\begin{aligned}
			\bbZ_i^\pi(\theta; \nu) :&= \fa{\bm{Z}_i^{j\br{i}, \pi}(\theta; \nu), \bm{Z}_i^{j\br{i}+1, \pi}(\theta; \nu), \cdots, \bm{Z}_i^{N, \pi}(\theta; \nu)} \\
			&= \mu_i\br{V_0, V_1^\pi(\theta), \cdots, V_{i}^\pi(\theta), X_{i}^\pi(\theta); \nu} \in  \bbR^{L\times \br{N - j(i) + 1}},\\
			\widetilde{\bbZ}_i^\pi (\theta; \nu)&:= \fa{\widetilde{\bm{Z}}_i^{j\br{i}, \pi}(\theta; \nu), \widetilde{\bm{Z}}_i^{j\br{i}+1, \pi}(\theta; \nu), \cdots, \widetilde{\bm{Z}}_i^{N, \pi}(\theta)} \\
			& = \phi_i\br{V_0, V_1^\pi(\theta), \cdots, V_{i}^\pi(\theta), X_{i}^\pi(\theta); \nu} \in  \bbR^{L\times \br{N - j(i) + 1}},
		\end{aligned}
	\end{equation}
	where $\mu_i\br{\cdots;\nu}, \phi_i\br{\cdots;\nu} \in \calN\calN_{i+1, L\times\br{N - j(i)+1}}$ are neural networks with input dimension $i + 1$, output dimension $L \times \br{N - j(i) + 1}$ and $\nu$ representing the set of trainable weights. Using this ansatz, we obtain:
	\begin{align}\label{eq: forward Euler BSDE}
		\bbY_{i+1}^\pi(\theta; \nu) = (1 + rh)\bbY_i^{\pi}(\theta;\nu) + \tilde{\bbZ}^{\pi}_i(\theta; \nu)\Delta W_{t_i} + \bbZ^\pi_i(\theta;\nu)\Delta W_{t_i}^\perp.
	\end{align}	
	By matching the terminal conditions of the forward SDE and the BSDE, we define the loss function as:
	\begin{align}\label{eq: loss}
		\calL\br{\theta; \nu} := \frac{1}{LN}\sum_{j= 1}^N\bbE\fa{\abs{\boldsymbol{G}\br{\e^{X^\pi _{T_j}\br{\theta}}} - \bm{Y}_{T_j}^{j, \pi}\br{\theta; \nu}}^2},
	\end{align}
	where $|x|$ denotes the Euclidean norm of $x \in \bbR^L$ and 
	\begin{align*}
		\boldsymbol{G}\br{\e^{X^\pi_{T_j}\br{\theta}}} := \fa{\br{\e^{X_{T_j}^\pi\br{\theta} +rT_j} - K_1}^+, \cdots, \br{\e^{X_{T_j}^\pi\br{\theta} +rT_j} - K_L}^+}^\top.
	\end{align*} 
	
	\subsection{Calibration task}
	The calibration objective is to identify the model parameters $\theta$ that minimize the discrepancy between the model-implied option prices $\bbY_0\br{\theta} = [\bm{Y}_0^{1}\br{\theta}, \cdots, \bm{Y}_0^{N}\br{\theta}] \in \bbR^{L \times N}$ and the market prices $\bbP_{MKT} = [(P_{MKT}(K_\ell, T_j))]_{L\times N} \in \bbR^{L \times N}$. This discrepancy is quantified by the objective function
	\begin{align*}
		F\br{\theta} := \frac{1}{LN}\left\|\bbY_0\br{\theta} - \bbP_{MKT}\right\|_F^2,
	\end{align*}
	where $\|\cdot\|_F$ denotes the Frobenius norm of a matrix. Given \eqref{eq:initial-bsde}, the initial numerical value $\mathbb{Y}_0^\pi$ coincides with $\mathbb{P}_{MKT}$, allowing us to reformulate $F(\theta)$ as:
	\begin{align}\label{eq: objective 2}
		F(\theta) = \frac{1}{LN}\left\|\mathbb{Y}_0(\theta) - \mathbb{Y}_0^{\pi}\right\|_F^2.
	\end{align}
	Hence, the calibration task \eqref{eq:objective1} is equivalent to the following optimization problem
	\begin{align*}
		\inf\limits_{{\theta \in \Theta}} F\br{\theta},
	\end{align*}
	where $\Theta$ is a suitable parametric space. Unlike traditional calibration schemes that rely on an explicit pricing map for $\mathbb{Y}_0(\theta)$, Theorem \ref{thm: thm1} establishes that $F(\theta)$ is uniformly bounded by $\mathcal{L}(\theta;\nu)$. Therefore, we reformulate the calibration procedure as the following joint optimization problem
	\begin{align}\label{eq: loss-total}
		\min\limits_{\substack{\theta \in \Theta\\\mu_i(\cdot;\nu), \phi_i(\cdot;\nu) \in \calN\calN}}\calL\br{\theta;\nu}.
	\end{align}
	The proposed calibration scheme is summarized in Algorithm \ref{alg: calibration_1}.
	
	\begin{algorithm}	
		\caption{Deep BSDE-based calibration scheme}
		\label{alg: calibration_1}
		\begin{algorithmic}	
			\Require{Time grid $0 = t_0 < t_1 < \cdots < t_n = T$; \\
				Market parameter $\zeta$;\\
				Market data $\bbP_{MKT}$; \\
				(Adaptive) learning rate $\alpha$; \\
				Initial guess $\theta_0$;}	          
			\State{\textbf{Initialization}:
				$\theta = \theta_0$;}
			\While{\text{not converge}}
			\State{Generate sample paths $\br{X_{i}^\pi(\theta), V_{i}^\pi(\theta), \Delta W_{t_i}, \Delta W_{t_i}^\perp}_{i = 1, \cdots, n}$ by $\eqref{eq: euler forward SDE}$}
			\State{Set $\bbY_0^\pi$ by \eqref{eq:initial-bsde}}
			\For{$i = 0, \cdots, n-1$}
			\State{Approximate $\mathbb{Z}_i^\pi(\theta;\nu)$ by $\mu_i(V_0, V_{1}^\pi(\theta),\cdots, V_{i}^\pi(\theta), X_{i}^\pi(\theta); \nu)$} 
			\State{Approximate $\widetilde{\mathbb{Z}}_i^\pi(\theta;\nu)$ by $\phi_i(V_0, V_{1}^\pi(\theta),\cdots, V_{i}^\pi(\theta), X_{i}^\pi(\theta); \nu)$}				
			\State{Get $\mathbb{Y}_{i+1}^\pi(\theta; \nu)$ by \eqref{eq: forward Euler BSDE}}
			\EndFor
			\State{Evaluate the loss function $\calL(\theta;\nu)$}		
			\State{Update $\theta \gets \theta - \alpha\nabla_{\theta}\calL$} 
			\State{Update $\nu \gets \nu -\alpha\nabla_{\nu}\calL$}	
			\EndWhile
			\State \Return{$\theta$}	
		\end{algorithmic}
	\end{algorithm}
	
	\section{Convergence analysis} \label{sec: convergence}
	In this section, we derive estimates for the target function $F\br{\theta} $ and the loss $\calL\br{\theta;\nu}$. For simplicity, we restrict our analysis to the case where $L = N = 1$ and $\zeta = {K, T}$. Recall that equations \eqref{eq: euler forward SDE}, \eqref{eq:initial-bsde}, \eqref{eq: NN approx} and \eqref{eq: forward Euler BSDE} yield the following discrete scheme for $i = 0, \cdots, n-1$,
	\begin{equation}\label{eq: euler FBSDE}
		\begin{aligned}		
			X_{t_{i+1}}^\pi(\theta) &= X_{t_i}^\pi(\theta) - \frac{1}{2}V_{t_i}^\pi(\theta) h + \rho\sqrt{V_{t_i}^\pi(\theta)}\Delta W_{t_i} + \sqrt{\br{1 - \rho^2}V_{t_i}^\pi(\theta)}\Delta W_{t_i}^\perp,\quad X_0^\pi = x_0, \\		
			Z_{t_i}^\pi(\theta;\nu) &= \mu_i\br{V_0, V_{t_1}^\pi(\theta),\cdots, V_{t_i}^\pi(\theta), X_{t_i}^\pi(\theta); \nu}, \\ 
			\tilde{Z}_{t_i}^\pi (\theta;\nu)&= \phi_i\br{V_0, V_{t_1}^\pi(\theta),\cdots, V_{t_i}^\pi(\theta), X_{t_i}^\pi(\theta); \nu},\\
			Y_{t_{i+1}}^\pi(\theta;\nu) &= \br{1 + rh}Y_{t_i}^\pi(\theta;\nu) + \tilde{Z}_{t_i}^\pi(\theta;\nu) \Delta W_{t_i} + Z_{t_i}^\pi (\theta;\nu)\Delta W_{t_i}^\perp,\quad Y_0^\pi(\zeta) = P_{MKT}.		
		\end{aligned} 
	\end{equation}
	Observe that $Y_{t_i}^\pi(\theta;\nu), Z_{t_i}^\pi(\theta;\nu), \tilde{Z}_{t_i}^\pi (\theta;\nu)$ are $\mathcal{F}_{t_i}$-measurable. In the sequel, let $C\in (0, \infty)$ denote a generic constant independent of $h$, the value of which may vary from line to line. We define the model parameter space $\Theta = \hua{\theta = (\theta_1, \cdots, \theta_n) \in \mathbb{R}^n | \theta_i \in [\theta_i^l, \theta_i^u], 1 \leq i \leq n}$ as a bounded set, where $n \in \mathbb{N}$ depends on the specific model and parameterization of model parameters. For any $0 \leq t_1 \leq t_2 \leq T$, \cite[Remark A.1]{bayer2022pricing} establishes that 
	\begin{align}\label{eq: regularity of V 2}
		\bbE\left[\abs{V_{t_2}(\theta) - V_{t_1}(\theta)}\right] + \bbE\fa{\int_{t_1}^{t_2} V_s(\theta) \de s} + \bbE\fa{\br{\int_{t_1}^{t_2} V_s(\theta) \de s}^2} \leq f(\abs{t_2 - t_1}),
	\end{align}
	where $f(t) = Ct^H$ for some constant $C$, reflecting the path properties of Volterra processes with a fractional kernel. Under Assumption \ref{assump: variance process} and $\eqref{eq: regularity of V 2}$, the strong solution to the SDE $\eqref{eq: log price}$ satisfies \cite[Equation A.1]{bayer2022pricing}
	\begin{align}\label{eq: boundedness of X}
		\bbE\fa{\sup_{t \in [0, T]}\abs{X_t(\theta)}^2} &\leq C\br{1 + x_0^2}.
	\end{align}
	Futhermore, the Euler-Maruyama approximation given by $\eqref{eq: euler forward SDE}$ has the following error estimate \cite[Equation A.2]{bayer2022pricing}
	\begin{align}\label{eq: err of X}
		\max_{i = 0, \cdots, n-1}\bbE\fa{\abs{X_{t_{i+1}}(\theta) - X_{t_{i+1}}^\pi(\theta)}^2 + \sup_{t \in [t_i, t_{i+1}]} \abs{X_t(\theta) - X_{t_i}^\pi(\theta)}^2} &\leq Cf(h).
	\end{align}
	Assumptions \ref{assump: variance process}, \ref{assump: payoff} and \eqref{eq: regularity of V 2} imply the existence and uniqueness of an adapted $L^2$-solution $(Y, Z, \tilde{Z})$ to the BSDE $\eqref{eq: BSDE}$. This, combined with $\eqref{eq: boundedness of X}$ establishes $L^2$-regularity result of $Y$ \cite[Equation A.4]{bayer2022pricing}:
	\begin{align}
		\max_{i = 0, \cdots, n-1} \bbE\fa{\sup_{t \in [t_i, t_{i+1})}\abs{Y_t(\theta) - Y_{t_i}(\theta)}^2} \leq  C_1h, \label{eq: L^2 of Y}
	\end{align}
	for a constant $C_1 \in (0, \infty)$. In the following theorem, we aim to bound the numerical error of the BSDE solutions, including the target $F(\theta) = (Y_0(\theta) - Y_0^\pi(\theta))^2$, in terms of the temporal discretization error, the numerical error of the payoff and the loss function
	\begin{align}\label{eq: loss_2}
		\calL(\theta;\nu) =\bbE\fa{\abs{G\br{\e^{X_T^\pi(\theta)}} - Y_T^\pi(\theta;\nu)}}^2.
	\end{align}
	\begin{theorem}\label{thm: thm1}
		Suppose Assumptions \ref{assump: variance process} and \ref{assump: payoff} hold. Then there exists  $C>0$, depending on $r$, $T$ and $C_1$ but independent of $h$ such that 
		\begin{equation}\label{eq: simulation error}
			\begin{aligned}
				&\quad\sup_{t \in [0, T]}\bbE\fa{\abs{Y_t(\theta) - Y_t^\pi(\theta;\nu)}^2} + \bbE\fa{\int_0^T\abs{Z_t(\theta) - Z_t^\pi(\theta;\nu)}^2\de t} + \bbE\fa{\int_0^T\abs{\tilde{Z}_t(\theta) - \tilde{Z}_t^\pi(\theta;\nu)}^2 \de t} \\ &\leq C\br{h + \bbE\fa{\abs{G\br{\e^{X_T(\theta)}} - G\br{\e ^{X_T^\pi(\theta)}}}^2} + \calL(\theta; \nu)},			
			\end{aligned}
		\end{equation}
		where $Y_t^\pi(\theta;\nu) = Y_{t_i}^\pi(\theta;\nu)$, $Z_t^\pi(\theta;\nu) = Z_{t_i}^\pi(\theta;\nu)$ and $\tilde{Z}_t^\pi(\theta;\nu) = \tilde{Z}_{t_i}^\pi(\theta;\nu)$ for $t \in [t_i, t_{i+1}), i = 0, \cdots, n-1$. 	
	\end{theorem}
	\begin{proof}
		For the sake of notational brevity, we suppress the dependence on $\theta$ and $\nu$ throughout the proof. We begin by expressing the BSDE of $\eqref{eq: FBSDE}$ in the integral form
		\begin{align*}
			Y_{t_i} = Y_{t_{i+1}} - r\int_{t_i}^{t_{i+1}} Y_t \de t - \int_{t_i}^{t_{i+1}}\tilde{Z}_t \de W_t - \int_{t_i}^{t_{i+1}} Z_t\de W_t^\perp.
		\end{align*}
		The last equation of $\eqref{eq: euler FBSDE}$ gives 
		\begin{align*}
			Y_{t_i}^\pi = Y_{t_{i+1}}^\pi -rY_{t_i}^\pi h - \tilde{Z}_{t_i}^\pi \Delta W_{t_i} - Z_{t_i}^\pi \Delta W_{t_i}^\perp.
		\end{align*}
		Subtracting the latter from the former yields: 		
		\begin{align*}
			Y_{t_i} - Y_{t_i}^\pi = Y_{t_{i+1}} - Y^\pi_{t_{i+1}} - r\int_{t_i}^{t_{i+1}}\br{Y_t - Y^\pi_{t_i}}\de t - \int_{t_i}^{t_{i+1}}\br{\tilde{Z}_t - \tilde{Z}_{t_i}^\pi}\de W_t -\int_{t_i}^{t_{i+1}}\br{Z_t - Z_{t_i}^\pi}\de W_t^\perp.
		\end{align*}
		Substituting the decomposition $Y_t - Y_{t_i}^\pi = Y_t - Y_{t_i} + Y_{t_i} - Y_{t_i}^\pi$ into the equation above, we obtain 
		\begin{equation}
			\begin{aligned}\label{eq: splitting}		
				(1 + rh)\br{Y_{t_i} - Y^\pi_{t_i}} &= Y_{t_{i+1}} - Y^\pi_{t_{i+1}} - r\int_{t_i}^{t_{i+1}}\br{Y_t - Y_{t_i}}\de t  \\ 
				& - \int_{t_i}^{t_{i+1}}\br{\tilde{Z}_t - \tilde{Z}^\pi_{t_i}}\de W_t -\int_{t_i}^{t_{i+1}}\br{Z_t - Z^\pi_{t_i}}\de W_t^\perp .
			\end{aligned}
		\end{equation}
		Applying conditional expectation $\bbE\fa{\cdot | \calF_{t_i}}$ to both sides yields
		\begin{align*}
			\br{1 + rh}\br{Y_{t_i} - Y^{\pi}_{t_i}} = \bbE\fa{Y_{t_{i+1}} - Y^{\pi}_{t_{i+1}}\bigg|\calF_{t_i}} - r\bbE\fa{\int_{t_i}^{t_{i+1}} \br{Y_t - Y_{t_i}}\de t \bigg|\calF_{t_i}}.
		\end{align*}
		Using Young's inequality in the form 
		\begin{align}\label{eq:young}
			(a + b)^2 \leq (1 + \gamma)a^2 + \br{1 + \gamma^{-1}}b^2\;\quad\forall a,\, b,\, \gamma \in \bbR,\; \gamma > 0,
		\end{align}
		along with the Cauchy-Schwarz inequality and the $L^2$-regularity of $Y$ given in $\eqref{eq: L^2 of Y}$, we derive 
		\begin{align*}        
			&\quad(1 + rh)^2\bbE\fa{\abs{Y_{t_i} - Y^\pi_{t_i}}^2} \\
			&\leq (1 + \gamma)\bbE\fa{\br{\bbE\fa{Y_{t_{i+1}} - Y^\pi_{t_{i+1}} | \calF_{t_i}}}^2} + \br{1 + \frac{1}{\gamma}}r^2\bbE\fa{\br{\bbE\fa{\int_{t_i}^{t_{i+1}} \br{Y_t - Y_{t_i}}\de t \bigg| \calF_{t_i}}}^2} \\
			&\leq (1 + \gamma)\bbE\left[\abs{Y_{t_{i+1}} - Y^\pi_{t_{i+1}}}^2\right]+ \br{1+ \frac{1}{\gamma}}r^2\bbE\fa{\abs{\int_{t_i}^{t_{i+1}} \br{Y_t - Y_{t_i}}\de t}^2} \\
			&\leq (1 + \gamma)\bbE\left[\abs{Y_{t_{i+1}} - Y^\pi_{t_{i+1}}}^2\right] + C_1\br{1 + \frac{1}{\gamma}}r^2h^3.
		\end{align*}   
		Setting $\gamma := rh(rh + 2)$, we obtain
		\begin{align*}
			\bbE\fa{\abs{Y_{t_i} - Y^\pi_{t_i}}^2} \leq \bbE\left[\abs{Y_{t_{i+1}} - Y^\pi_{t_{i+1}}}^2\right] + \frac{C_1rh^2}{rh + 2}.
		\end{align*}
		Proceeding by induction, for $i = 0, \cdots, n-1$, we have 
		\begin{align*}
			\bbE\left[\abs{Y_{t_i} - Y^\pi_{t_i}}^2\right] \leq \bbE\left[\abs{Y_T - Y_{T}^\pi}^2\right] + \frac{C_1rh^2(n-i)}{rh + 2}. 
		\end{align*}
		Combining this estimate with the $L^2$-regularity of $Y$ and the triangle inequality, we obtain
		\begin{equation}\label{eq: bound for Y}
			\begin{aligned}
				&\quad\sup_{t \in [0, T]}\bbE\fa{\abs{Y_t - Y^\pi_t}^2}
				=\bigg({\max_{i=0,\cdots,n-1}\sup_{t\in[t_i,t_{i+1})}\bbE\fa{\abs{Y_t - Y^\pi_{t_i}}^2}}\bigg) \vee \bbE\fa{\abs{Y_T - Y_T^\pi}^2}\\
				&\leq \bigg({2\max_{i=0,\cdots,n-1}\sup_{t\in[t_i,t_{i+1})}\bbE\fa{\abs{Y_t -Y_{t_i}}^2+\abs{Y_{t_i} -Y^\pi_{t_i}}^2}}\bigg) \vee \bbE\fa{\abs{Y_T - Y_T^\pi}^2}\\		
				&\leq \bigg({2\max_{i=0,\cdots,n-1}\Big(\bbE\Big[{\sup_{t\in[t_i,t_{i+1})}\abs{Y_t -Y_{t_i}}^2}\Big]+\bbE\fa{\abs{Y_{t_i} -Y^\pi_{t_i}}^2}
					\Big)}\bigg) \vee \bbE\fa{\abs{Y_T - Y_T^\pi}^2}\\
				& \leq 2C_1h + \max_{i=0,\cdots,n-1}\frac{2C_1rh^2(n-i)}{rh + 2} +2\bbE\left[\abs{Y_T - Y_T^\pi}^2\right] \\
				& \leq 2(1+rT)C_1h+2\bbE\left[\abs{Y_T - Y_T^\pi}^2\right]. 
			\end{aligned}
		\end{equation}
		Next, we estimate the error in the $Z$ and $\tilde{Z}$ components in $\eqref{eq: simulation error}$. From $\eqref{eq: splitting}$, we have
		\begin{align*}
			&\quad Y_{t_{i+1}} - Y^\pi_{t_{i+1}} - r\int_{t_i}^{t_{i+1}}\br{Y_t - Y_{t_i}}\de t \\		
			&= \br{1 + rh}\br{Y_{t_i} - Y^\pi_{t_i}} + \int_{t_i}^{t_{i+1}}\br{\tilde{Z}_t - \tilde{Z}^\pi_{t_i}}\de W_t + \int_{t_i}^{t_{i+1}}\br{Z_t - Z^\pi_{t_i}}\de W_t^\perp.
		\end{align*}
		Squaring both sides and taking conditional expectations $\bbE\fa{\cdot | \calF_{t_i}}$, we have     
		\begin{align*}
			&\quad\bbE\fa{\br{Y_{t_{i+1}} - Y^\pi_{t_{i+1}} - r\int_{t_i}^{t_{i+1}}\br{Y_t - Y_{t_i}}\de t}^2 \bigg | \calF_{t_i}} \\
			&= \br{1 + rh}^2\br{Y_{t_i} - Y^\pi_{t_i}}^2 + \bbE\fa{\int_{t_i}^{t_{i+1}}\br{\tilde{Z}_t - \tilde{Z}^\pi_{t_i}}^2 \de t\bigg | \calF_{t_i}} + \bbE\fa{\int_{t_i}^{t_{i+1}}\br{Z_t - Z^\pi_{t_i}}^2 \de t \bigg| \calF_{t_i}}.
		\end{align*} 
		Taking expectations and using tower property gives
		\begin{equation}\label{eq: Y_Z}
			\begin{aligned}
				&\quad\bbE\fa{\br{Y_{t_{i+1}} - Y^\pi_{t_{i+1}} - r\int_{t_i}^{t_{i+1}}\br{Y_t - Y_{t_i}}\de t}^2} \\
				&= \br{1 + rh}^2\bbE\fa{\abs{Y_{t_i} - Y^\pi_{t_i}}^2} + \bbE\fa{\int_{t_i}^{t_{i+1}}\br{\tilde{Z}_t - \tilde{Z}^\pi_{t_i}}^2\de t} + \bbE\fa{\int_{t_i}^{t_{i+1}}\br{Z_t - Z^\pi_{t_i}}^2 \de t}.
			\end{aligned}
		\end{equation}
		With the Young's inequality \eqref{eq:young}, for any $\lambda>0$, we have 
		\begin{align*}
			&\quad\bbE\fa{\int_{t_i}^{t_{i+1}}\br{\tilde{Z}_t - \tilde{Z}^\pi_{t_i}}^2\de t} + \bbE\fa{\int_{t_i}^{t_{i+1}}\br{Z_t - Z^\pi_{t_i}}^2 \de t} \\
			&\leq \br{1 + \lambda}\bbE\fa{\abs{Y_{t_{i+1}} - Y^\pi_{t_{i+1}}}^2} + \br{1 + \frac{1}{\lambda}}r^2\bbE\fa{\br{\int_{t_i}^{t_{i+1}}\br{Y_t - Y_{t_i}}\de t}^2} - \br{1 + rh}^2\bbE\fa{\abs{Y_{t_i} - Y^\pi_{t_i}}^2}\\
			& \leq \br{1 + \lambda}\bbE\fa{\abs{Y_{t_{i+1}} - Y^\pi_{t_{i+1}}}^2} - \br{1 + rh}^2\bbE\fa{\abs{Y_{t_i} - Y^\pi_{t_i}}^2} 
			+\br{1 + \frac{1}{\lambda}}C_1r^2h^3. 
		\end{align*}
		We take $\lambda = rh\br{rh + 2}$ and sum the inequalities over $i = 0, \cdots, n-1$: 		
		\begin{equation}\label{eq: bound for Z}
			\begin{aligned}
				&\quad\bbE\fa{\int_0^T\br{\tilde{Z}_t - \tilde{Z}_t^\pi}^2 \de t} + \bbE\fa{\int_0^T\br{Z_t - Z_{t}^\pi}^2 \de t}\\
				&\leq \br{1 + rh}^2\br{\bbE\fa{\abs{Y_T - Y_T^\pi}^2} - \bbE\fa{\abs{Y_0 - Y^\pi_0}^2}} + \br{1 + rh}^2\frac{C_1rhT}{rh + 2} \\
				&\leq \br{1 + rT}^2\br{C_1rhT + \bbE\fa{\abs{Y_T - Y_T^\pi}^2}}.
			\end{aligned}
		\end{equation}	
		Finally, given that $Y_T = G\br{e^{X_T}}$, we decompose the terminal misfit as 
		\begin{equation}\label{eq: terminal misfit}
			\begin{aligned}
				\bbE\fa{\abs{G\br{\e^{X_T}} - Y_T^\pi}^2} &\leq 2\bbE\fa{\abs{G\br{\e^{X_T}} - G\br{\e ^{X_T^\pi}}}^2} + 2\bbE\fa{\abs{G\br{\e ^{X_T^\pi}} - Y_T^\pi}^2} \\
				&\leq C\br{\bbE\fa{\abs{G\br{\e^{X_T}} - G\br{\e ^{X_T^\pi}}}^2} + \calL(\theta;\nu)},
			\end{aligned}
		\end{equation}
		We complete the proof by combining estimates $\eqref{eq: bound for Y}, \eqref{eq: bound for Z}$ and $\eqref{eq: terminal misfit}$. 
	\end{proof}
	
	Next, we demonstrate that the loss \eqref{eq: loss_2} can be arbitarily small given a suitable set of model parameters and the  random function approximation capability of neural networks (see \cite[Proposition 4.2]{bayer2022pricing} for details). We denote the $L^2$-regularity of the pair $(Z, \tilde{Z})$ as follows:
	\begin{align*}
		\varepsilon^Z\br{h} &:= \bbE\fa{\sum_{i = 0}^{n-1}\int_{t_i}^{t_{i+1}}\abs{Z_t(\theta) - \bar{Z}_{t_i}(\theta)}^2\de t},\quad \bar{Z}_{t_i}(\theta) := \frac{1}{h}\bbE\fa{\int_{t_i}^{t_{i+1 }}Z_t(\theta) \de t \bigg| \calF_{t_i}},\\
		\varepsilon^{\tilde{Z}}\br{h} &:= \bbE\fa{\sum_{i = 0}^{n-1}\int_{t_i}^{t_{i+1}}\abs{\tilde{Z}_t(\theta) - \bar{\tilde{Z}}_{t_i}(\theta)}^2\de t},\quad    
		\bar{\tilde{Z}}_{t_i}(\theta) := \frac{1}{h}\bbE\fa{\int_{t_i}^{t_{i+1 }}\tilde{Z}_t(\theta) \de t \bigg| \calF_{t_i}}.
	\end{align*}
	Since $(\bar{Z}(\theta), \bar{\tilde{Z}}(\theta))$ is the $L^2$-projection of $(Z(\theta), \tilde{Z}(\theta))$, $\varepsilon^Z\br{h}$ and $\varepsilon^{\tilde{Z}}\br{h}$ vanish as $h \to 0$. 
	\begin{theorem}\label{thm: thm2}
		Suppose Assumption $\eqref{assump: payoff}$ holds and the parameter space $\Theta$ is compact. Then there exists $C>0$ that is depending on $r, T, \Theta$ but is independent of $h$ and $\theta$, such that for sufficiently small $h$,
		\begin{equation*}
			\begin{aligned}
				\inf_{\substack{\theta \in \Theta\\ \mu_i, \phi_i \in \calN\calN}}\calL(\theta;\nu) &\leq C\br{h + \sup _{\theta \in \Theta}\bbE\fa{\abs{G\br{\e^{X_T(\theta)}} - G\br{\e^{X_T^\pi(\theta)}}}^2} + \inf_{\theta \in \Theta}\abs{Y_0\br{\theta} - Y_0^\pi}^2 } \\
				&+ C\br{\varepsilon^Z\br{h} + \varepsilon^{\tilde{Z}}\br{h} + h\inf_{\mu_i, \phi_i \in \calN\calN}\sum_{i=0}^{n-1}\bbE\fa{\abs{\bar{Z}_{t_i}(\theta) - \mu_i}^2} + \bbE\fa{\abs{\bar{\tilde{Z}}_{t_i}(\theta) - \phi_i}^2}}.
			\end{aligned}
		\end{equation*}
		\begin{proof}
			For brevity, we remove the dependence on $\theta$ and $\nu$ in the proof. We first decompose the loss as
			\begin{equation}\label{eq: thm2 1}
				\begin{aligned}
					\bbE\fa{\abs{G\br{\e^{X_T^\pi}} - Y_T^\pi}^2} &\leq 2\bbE\fa{\abs{G\br{\e^{X_T}} - G\br{\e^{X_T^\pi}}}^2} + 2\bbE\fa{\abs{Y_T - Y_T^\pi}^2}.
				\end{aligned}
			\end{equation}
			By $\eqref{eq: Y_Z}$ and using Young's inequality of the form 
			\begin{align*}
				\br{a + b}^2 \geq (1 - h)a^2 + \br{1 - \frac{1}{h}}b^2 \geq \br{1 - h}a^2 - \frac{1}{h}b^2,
			\end{align*}
			we provide an estimate of $\bbE[ \abs{Y_T - Y_T^\pi}^2]$:
			\begin{align*}
				&\br{1 + rh}^2\bbE\fa{\abs{Y_{t_i} - Y_{t_i}^\pi}^2} + \bbE\fa{\int_{t_i}^{t_{i+1}}\br{\tilde{Z}_{t} - \tilde{Z}^\pi_{t_i}}^2 \de t} + \bbE\fa{\int_{t_i}^{t_{i+1}}\br{Z_{t} - Z_{t_i}^\pi}^2 \de t} \\
				&\geq \br{1 - h}\bbE\fa{\abs{Y_{t_{i+1}} - Y_{t_{i+1}}^\pi}^2} - \frac{1}{h}r^2\bbE\fa{\abs{\int_{t_i}^{t_{i+1}} \br{Y_t - Y_{t_i}}\de t}^2}.
			\end{align*}
			Then Cauchy's inequality and $\eqref{eq: L^2 of Y}$ indicate that 
			\begin{align*}
				&\bbE\fa{\abs{Y_{t_{i+1}} - Y_{t_{i+1}}^\pi}^2} \leq \frac{\br{1 + rh}^2}{1 - h}\bbE\fa{\abs{Y_{t_i} - Y_{t_i}^\pi}^2}\\
				&+ \frac{1}{1-h}\underbrace{\left(r^2C_1h^2 
					+ \bbE\fa{\int_{t_i}^{t_{i+1}}\br{\tilde{Z}_{t} - \tilde{Z}^\pi_{t_i}}^2 \de t} + \bbE\fa{\int_{t_i}^{t_{i+1}}\br{Z_{t} - Z_{t_i}^\pi}^2 \de t}\right)}_{=:g_i}.
			\end{align*}
			By the discrete Gr\"onwall inequality \cite[Proposition 3.2]{emmrich1999discrete}, for sufficiently small $h$ 
			\begin{align*}
				\bbE\fa{\abs{Y_T - Y_T^\pi}^2}
				&\leq \bigg({\frac{\br{1 + rh}^2}{1 - h}}\bigg)^n\bigg({\bbE\fa{\abs{Y_0 - Y_0^\pi}^2} + \frac{1}{1 -h}\sum_{j=0}^{n-1} \bigg({\frac{\br{1 + rh}^2}{1 - h}}\bigg)^{-\br{j+1}}g_j}\bigg)\\
				&\leq \bigg({\frac{\br{1 + rh}^2}{1 - h}}\bigg)^n\bigg({\bbE\fa{\abs{Y_0 - Y_0^\pi}^2} + \sum_{j = 0}^{n-1}g_j}\bigg)\\
				&\leq \e^{2\br{r+1}T}\Bigg(\bbE\fa{\abs{Y_0 - Y_0^\pi}^2} + r^2C_1hT + \sum_{j = 0}^{n-1}\bbE\fa{\int_{t_i}^{t_{i+1}}\br{\tilde{Z}_t - \tilde{Z}_{t_i}^\pi}^2\de t} \\
				&\hspace{5.5em}+ \sum_{j = 0}^{n-1}\bbE\fa{\int_{t_i}^{t_{i+1}}\br{Z_t - Z_{t_i}^\pi}^2\de t}\Bigg).
			\end{align*}
			Note that 
			\begin{align*}
				\sum_{i = 0}^{n-1}\bbE\fa{\int_{t_i}^{t_{i+1}}\br{Z_t - Z_{t_i}^\pi}^2\de t} &\leq 2\sum_{i = 0}^{n-1}\bbE\fa{\int_{t_i}^{t_{i+1}}\br{Z_t - \bar{Z}_{t_i}}^2\de t} 
				+ 2\sum_{i = 0}^{n-1}\bbE\fa{\int_{t_i}^{t_{i+1}}\br{\bar{Z}_{t_i} - Z_{t_i}^\pi}^2 \de t} \\
				&\leq 2\varepsilon^Z\br{h} + 2h\sum_{i=0}^{n-1}\bbE\fa{\abs{\bar{Z}_{t_i} - Z_{t_i}^\pi}^2}.
			\end{align*}
			Similarly, we have 
			\begin{align*}
				\sum_{i = 0}^{n-1}\bbE\fa{\int_{t_i}^{t_{i+1}}\br{\tilde{Z}_t - \tilde{Z}_{t_i}^\pi}^2\de t} \leq 2\varepsilon^{\tilde{Z}}\br{h} + 2h\sum_{i=0}^{n-1}\bbE\fa{\abs{\bar{\tilde{Z}}_{t_i} - \tilde{Z}_{t_i}^\pi}^2}. 
			\end{align*}
			Combining these estimates and taking infimum on both sides give the desired result.	
		\end{proof}
	\end{theorem}
	\section{Numerical experiments}\label{sec: exp}
	In this section, we evaluate the performance of the proposed calibration scheme within the rBergomi framework. We begin by assessing performance using simulated data. Because synthetic data is noise-free and devoid of market imperfections, it allows for an assessment of calibration accuracy independent of the rBergomi model's inherent limitations. Subsequently, we conduct calibrations using historical market data, treating model parameters as either scalars or time-varying functions. As the original mSOE scheme is ill-suited for time-varying parameters, we provide an adapted version in Appendix \ref{append: mSOE scheme for time-dependent rBergomi} to generate the necessary samples of $\br{X_t, V_t, W_t, W_t^\perp}$.
	
	We first outline the general experimental settings. We employ TensorFlow \cite{abadi2016tensorflow} as the automatic differentiation framework, largely adhering to the implementation of the \texttt{Deep BSDE Solver} \cite{deepBSDE}. 
	All experiments are conducted on the following computing infrastructure:
	CPU: Dual Intel Xeon 6226R (16 Core); GPU: NVIDIA Tesla V100 32GB SXM2; OS: Rocky Linux 8 (x86-64). The neural network architecture is defined as follows: for $i = 0, 1, \cdots, n-1$, the sub-networks $\mu_i$ and $\phi_i$ consist of fully connected neural networks with two hidden layers, each containing 32 neurons, utilizing a leaky ReLU activation function with a negative slope of 0.3. Batch normalization is applied immediately following each linear transformation and prior to activation. All trainable weights in the dense layers of $\mu_i$ and $\phi_i$ are initialized using the Xavier uniform distribution, with biases initialized to zero. No pre-training is employed.

	\subsection{Numerical accuracy} \label{subsec: numerical accuracy}
	In this subsection, we examine the calibration performance of the deep BSDE scheme using synthetic data. We utilize the ground truth model parameters $\theta$ and the initial guesses listed in Table \ref{tab: model param}, with $r = 0.03$ and $x_0 = 0$, to compute $\bbP_{MKT}$. All model parameters are treated as constants. This initial guess is common to all numerical experiments in this subsection.
	
	\begin{table}[htp]
		\centering
		\begin{tabular}{c cccc}
			\toprule
			& $\xi_0(t)$ & $H$ & $\rho$ & $\eta$ \\
			\hline
			Ground truth & 0.09 & 0.07 & -0.9 & 1.9\\
			Initial guess & 0.15 & 0.12 & -0.7 & 1.5\\
			\bottomrule
		\end{tabular}
		\caption{The set of ground truth parameters used for $\bbP_{MKT}$ and the initial guess.}
		\label{tab: model param}
	\end{table}
	The market parameter $\zeta$ is set to be 
	\begin{align*}
		\br{K_1, K_2, \cdots, K_5} &= \br{0.90, 0.95, 1.0, 1.05, 1.10},  \\ 
		\br{T_1, T_2, \cdots, T_5} &= \br{0.2, 0.4, 0.6, 0.8, 1}, 
	\end{align*}
	with $L = 5, N = 5$. $\bbP_{MKT}$ is computed using the mSOE scheme with a time step size $1/1000$ and $2^{14}$ Monte Carlo simulations. For the deep BSDE scheme, $2^7$ trajectories are generated per iteration, with a learning rate of 0.01. 
	
	Figure \ref{fig: error evolution} illustrates the evolution of three error metrics over the training iterations: $F(\theta)$, the average relative error, and the maximum relative error. Using the model parameters outputted by the calibration algorithm at each iteration, we compute $F(\theta)$ via $\eqref{eq: objective 2}$, where $\mathbb{Y}_0(\theta)$ is obtained using the mSOE scheme with a $1/1000$ step size and $2^{14}$ Monte Carlo simulations. The relative error is defined as	
	\begin{align}
		\text{Relative error} := \abs{\frac{\bbY_0(\theta) - \bbP_{MKT}}{\bbP_{MKT}}}\in \bbR^{L \times N}, \label{eq: rele err}
	\end{align}
	where $\abs{\cdot}$ denotes the absolute value, and the division is implemented element-wise. The average relative error corresponds to the mean of the entries in this matrix, while the maximum relative error represents the largest element. 
	\begin{figure}[htp]
		\includegraphics[width=1\linewidth]{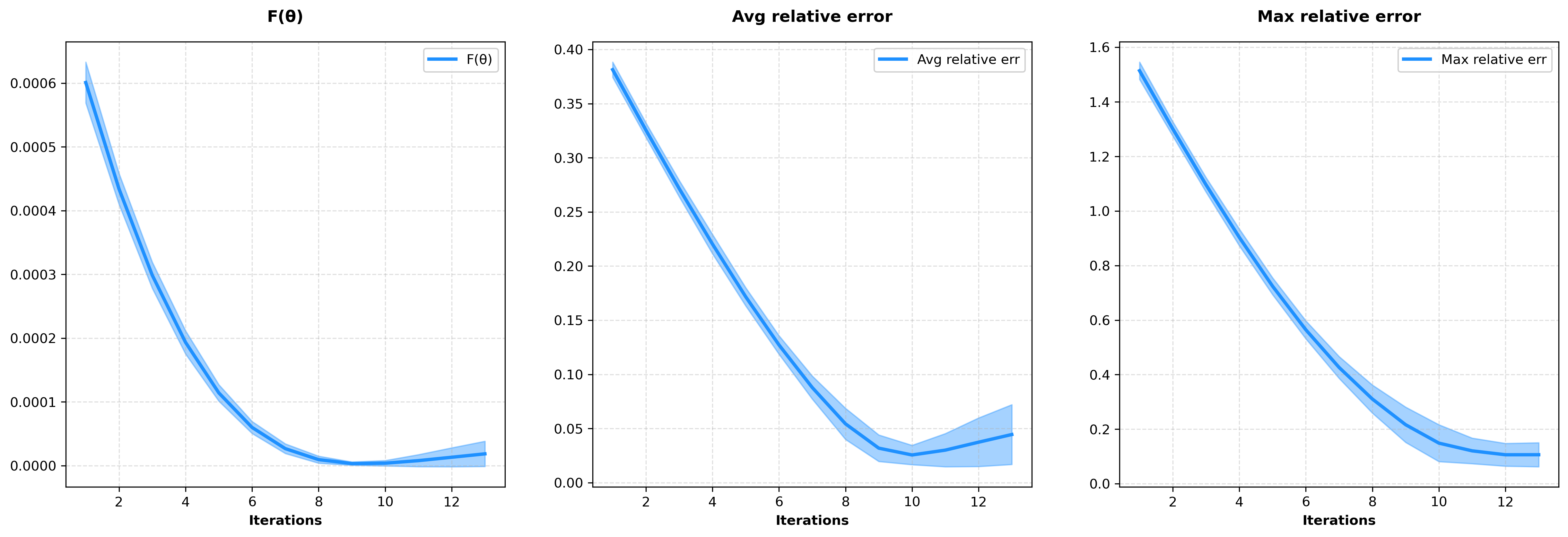}
		\centering
		\caption{Plot of error variation against the number of iterations using Algorithm \ref{alg: calibration_1}, with $h = 1/20$.}
		\label{fig: error evolution}
	\end{figure}
	In each panel, the solid blue curve represents the sample mean across ten independent runs using distinct random seeds, while the shaded region denotes one standard deviation around the mean. As shown in the left panel, $F(\theta)$ exhibits rapid decay within the first 10 iterations. The variance across random seeds is small at all iterations, which indicates highly stable optimization dynamics. This stability supports the implementation of an early stopping strategy once a desirable set of model parameters is identified. The middle panel shows that the average relative error decreases consistently from roughly 0.38 to 0.03 within the first 10 iterations, which demonstrates that the deep BSDE calibration method can accurately reconstruct the pricing surface with relatively few optimization steps. The right panel reveals a similar downward trend for the maximum relative error. However, the variability across seeds is larger than that of the average error. This suggests that the maximum relative error is driven by specific and challenging strike-maturity pairs that are sensitive to the optimization process. This observation is confirmed by the heatmaps in Figure \ref{fig: heatmap}.
	
	\begin{figure}[htp]
		\includegraphics[width=1.05\linewidth]{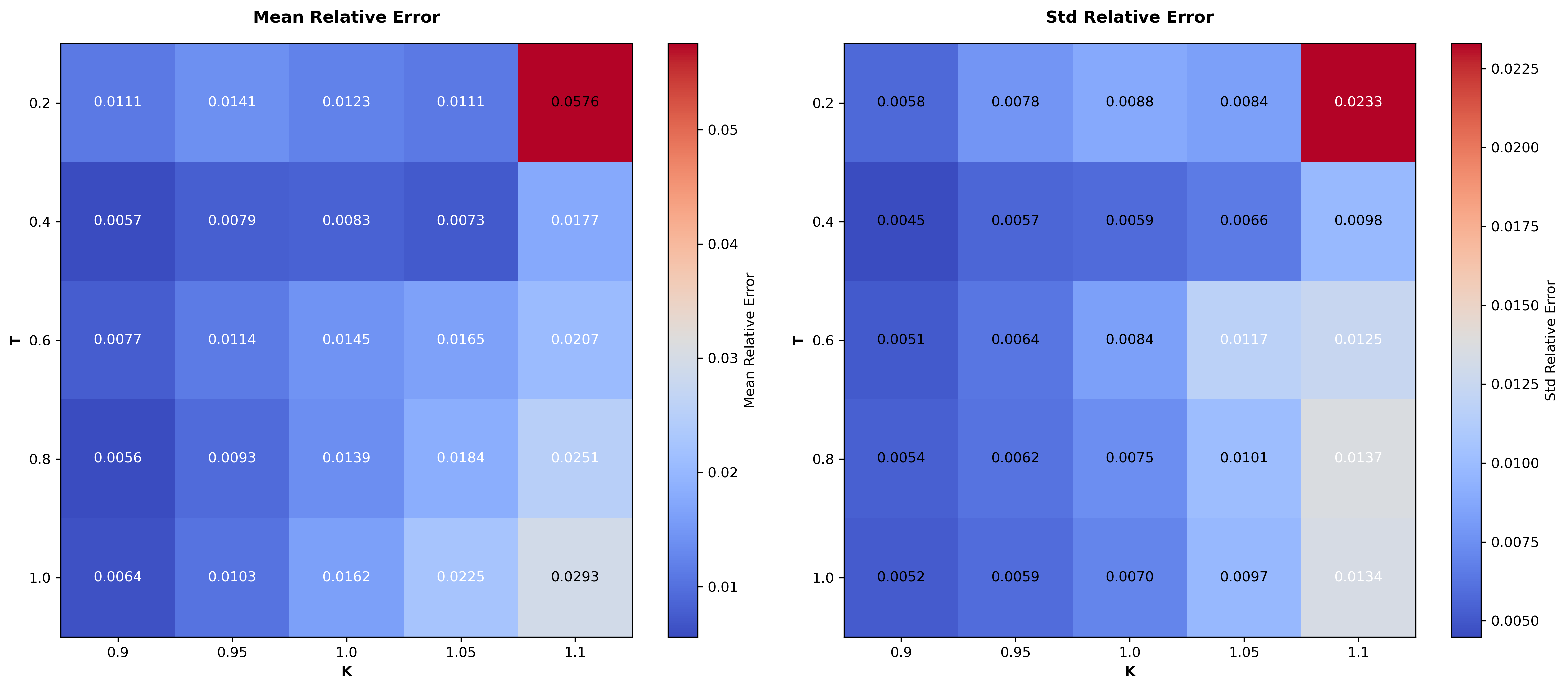}
		\centering
		\caption{The heatmaps of the maean relative pricing error and the standard deviation of relative pricing error, with $h = 1/20$.}
		\label{fig: heatmap}
	\end{figure}
	Figure \ref{fig: heatmap} displays the mean relative pricing error (left panel) and the standard deviation of the relative pricing error (right panel) across different strike levels $K$ and maturities $T$, based on 10 independent runs. Across most of the grid, the mean relative error remains in the range of $0.5\% - 2\%$, indicating that the parameters calibrated by the deep BSDE scheme achieve consistently accurate pricing over the strike-maturity grid. For a fixed maturity, the pricing error increases as $K$ moves from in-the-money to out-of-the-money. Elevated errors appear in the region of high strikes and short maturities, where the payoff corresponds to an extremely low probability mass. This results in noisy Monte Carlo-based gradients and increased variance across runs.
	
	Table \ref{tab: calibration_1} reports the impact of the time discretization size $h$ on the performance and stability of the deep BSDE calibration scheme. Each entry represents the mean of 10 independent training runs using different random seeds. The algorithm is terminated early when the target $F(\theta)$ does not decrease over a certain number of iterations, and we refer this number as "Patience". The table indicates that as the time step size $h$ decreases, the computational burden increases as evidenced by the second and third columns. Their product gives an estimate for the total training time. Let $\theta^\ast$ be the final output of the calibration algorithm. The mean squared error $F(\theta)$ decreases approximately proportionally with the refinement of $h$, which yields a more accurate approximation to the true solution. The average relative error also decreases steadily with a smaller $h$, implying the improvement for the global accuracy. The reduction in the maximum relative error is not monotonic. A finer time grid may degrade the worst-cast accuracy due to the optimization difficulty.
	\begin{table}[htp]
		\centering	
		\begin{adjustbox}{max width=\textwidth}
			\begin{tabular}{ccc |cc|cc|cc|c}
				\toprule			
				\makecell{$h$} & \makecell{Training \\ (s/Iter)}& \makecell{$\#$ \\ Iters} & $F(\theta^\ast)$ & $\text{SD}$& \makecell{Avg Rel. \\error} & $\text{SD}$ &\makecell{Max Rel. \\error} & $\text{SD}$ & $\calL(\theta^\ast;\nu^\ast)$ \\ \hline \\
				$1/10$ & 12.8426 &10.5 & 1.2375e-5& 2.4428e-6& 0.0370& 0.0049& 0.0722& 0.0130 & 0.0072\\
				$1/20$  &16.7639& 12.8& 8.0759e-6& 2.5289e-6& 0.0294& 0.0050& 0.0607& 0.0123& 0.0060\\
				$1/40$  &19.5850& 23.8& 4.3511e-6& 1.4783e-6& 0.0223& 0.0024& 0.0870& 0.0285& 0.0018\\				
				\bottomrule
			\end{tabular}
		\end{adjustbox}
		\caption{Calibration error variation of Algorithm \ref{alg: calibration_1} with step size $h$, with Patience = 1.}  
		\label{tab: calibration_1}	
	\end{table}

	\subsection{Calibration with historical data} \label{subsec: calibration by hist}	
	We calibrate the model to S\&P 500 index (ticker: SPX) European call option data sourced from OptionMetrics \url{www.optionmetrics.com} via Wharton Research Data Services (WRDS). We designate the underlying price $S_0$ as the official closing price of the S\&P 500 index and proxy the spot variance $V_0$ using the square of the at-the-money implied volatility with the shortest maturity. These values are listed in Table \ref{tab: price and variance}. The option prices used are the average of the maximum bid and minimum ask prices. Calibration is based on the selected strikes and maturities listed in Table \ref{tab: train and test set}. Test Set 1 refers to strike-wise out-of-sample testing, aiming to assess the calibrated model's pricing accuracy on unseen strikes. Test Set 2 refers to temporal out-of-sample testing, using parameters calibrated on February 28th to predict option prices on March 1st. 
	
	\begin{table}[htp]
		\centering	
		\begin{adjustbox}{max width=\textwidth}
			\begin{tabular}{c|c|c}
				\toprule
				& \text{Feb 28th, 2023} & \text{Mar 1st, 2023} \\
				\hline
				$S_0$ & 3970.15& 3951.39\\
				\hline
				$V_0$ & 0.0378& 0.0354\\		
				\bottomrule
			\end{tabular}
		\end{adjustbox}
		\caption{Underlying price $S_0$ and spot variance $V_0$ used for calibration.} 
		\label{tab: price and variance}
	\end{table}
	\begin{table}[htp]
		\centering    
		\begin{adjustbox}{max width=\textwidth}
			\begin{tabular}{c|c|c|c}
				\toprule
				& \text{Date} & $K$ & $T$ \\
				\hline
				\text{Train set} & \text{Feb 28th, 2023}& $[3750, 3800, 3850, 3900, 3950, 4000, 4050, 4100, 4150]$ & \multirow{3}{*}{$[0.2, 0.4, 0.6, 0.8, 1.0]$}\\
				\cline{1-3}
				\text{Test set 1}  & \text{Feb 28th, 2023} & $[3775, 3875, 3975, 4075, 4175]$ & \\
				\cline{1-3}
				\text{Test set 2} &\text{Mar 1st, 2023} & $[3750, 3850, 3950, 4050, 4150]$ & \\
				\bottomrule
			\end{tabular}
		\end{adjustbox}
		\caption{Strikes and maturities of train set and test set.} 
		\label{tab: train and test set}
	\end{table}
	
	Since the raw option data are not provided in a grid format corresponding to the maturities listed in Table \ref{tab: train and test set}, we select $\bbP_{MKT}$ such that the corresponding maturities $(\tilde{T}_1, \cdots, \tilde{T}_5)$ are close to $\br{T_1, \cdots, T_5}$ element-wise. To obtain $\bbY_0^\pi \in \bbR^{L \times N}$ aligned with the maturity grid $\br{T_1, \cdots, T_5}$, we apply cubic spline interpolation and constant extrapolation to $\bbP_{MKT}$ for each $K_{\ell}$ with the \texttt{tf-quant-finance} library \cite{tfquantfinance}. To compute $F(\theta)$, we apply interpolation to $\bbY_0(\theta)$ to ensure it aligns with the maturities associated with $\bbP_{MKT}$ for consistency. 
	
	Subsequently, we allow the interest rate $r$ to be time-varying and treat it as a function of time, $r = r(s)$. We define the corresponding log-price as 
	$ X_s := -\int_0^s r(u)\de u + \ln S_s$,
	which also follows $\eqref{eq: log price}$. The European option price with payoff function $h(\cdot)$ is given by
	\begin{align*}
		u(x) := \bbE\fa{\e^{-\int_0^T r(s)\de s}h\br{\exp\br{X_T + \int_0^Tr(s)\de s}}},
	\end{align*}
	and the associated BSDE is 
	\begin{align*}
		-\de Y_s &= - r(s)Y_s\de s - \tilde{Z}_s\de W_s - Z_s \de W_s^\perp,\\
		Y_T &= h\br{\exp\br{X_T + \int_0^T r(s) \de s}}.
	\end{align*}
	As the dataset contains current interest rate information at various maturities, we apply cubic spline interpolation to obtain $r(t)$ for any $t \in [0, T]$. The simulation scheme for the forward SDE remains unaffected by the time-varying interest rate. For the backward SDE, equation \eqref{eq: forward Euler BSDE} is modified to 
	\begin{align*}
		\bbY_{i+1}^\pi(\theta; \nu) = (1 + r(t_i)h)\bbY_i^{\pi}(\theta;\nu) + \tilde{\bbZ}^{\pi}_i(\theta; \nu)\Delta W_{t_i} + \bbZ^\pi_i(\theta;\nu)\Delta W_{t_i}^\perp.
	\end{align*}	
	It is well known that calibration is not necessarily a convex optimization problem, and the target $F(\theta)$ often exhibits multiple local minima. Consequently, a suitable initial guess for the model parameters is crucial for calibration performance. Motivated by the concept of using neural networks for warm-starting \cite{zhou2023neural}, one approach involves training the following inverse mapping using synthetic data: 
	\begin{align*}
		\Psi^{-1}: \mathbb{R}^{L\times N} \to \Theta,\quad\tilde{\calP}(\theta) \mapsto \theta,
	\end{align*}	
	and regarding $\Psi^{-1}(\bbP_{MKT})$ as the initial guess. In our implementation, we utilize the output of Algorithm \ref{alg: calibration_1} trained on a small subset of market data to obtain a rapid initialization. To perform calibration for subsequent dates, the calibrated results from the preceding day serve as a reasonable initial guess by assuming market conditions do not change drastically over short periods. More elaborate methods for determining initial guesses are left for future research. All numerical experiments in this section share the same initial guess of model parameters and parametric space $\Theta$ listed in Table \ref{tab: initial guess}.
	\begin{table}[htp]
		\centering
		\begin{tabular}{c cccc}
			\toprule
			& $\xi_0(t)$ & $H$ & $\rho$ & $\eta$ \\
			\hline			
			Initial guess & 0.05 & 0.04 & -0.6 & 2.5\\
			Constraint & $[0.01, 0.2]$ & $[0.01, 0.499]$ & $[-0.999, -0.1]$ & $[1.0, 4.0]$ \\
			\bottomrule
		\end{tabular}
		\caption{The initial guess of model parameters and the parametric space $\Theta$.}
		\label{tab: initial guess}
	\end{table}
	We employ the following general experimental setting: $\mathbb{Y}_0(\theta)$ is obtained by the mSOE scheme with step size $1/1000$ and $2^{17}$ Monte Carlo repetitions. The BSDE was solved with $h = 1/20$ and the number of samples generated per iteration is $2^{13}$.

	\subsubsection{Model parameters as scalars}
	First, we consider the case that all model parameters are scalars. We present in-sample and out-of-sample errors in Table \ref{tab: calibration hist 1}, where the result presented is the mean of 10 independent runs with different random seeds. In a comparative analysis against the standard Monte Carlo benchmark, the deep BSDE scheme highlights a significant advantage in computational efficiency.  While the fully converged Monte Carlo method achieves the lowest in-sample average relative error $0.91\%$, it requires significantly more computation time. In contrast, the deep BSDE scheme achieves a highly competitive error rate $1.36\%$ within $240$s. When the Monte Carlo method is restricted to a similar time budget, its accuracy degrades substantially to $2.97\%$. Moreover, the deep BSDE scheme maintains low error rates for both the strike-wise and temporal-wise out-of-sample datasets, which demonstrates its robust generalization capability comparable to the converged benchmark. The calibrated model parameters of the best run out of 10 runs are presented in Table \ref{tab: calibrated model param 1}.
	
	\begin{table}[htp]
		\centering
		\begin{adjustbox}{max width=\textwidth}
			\begin{tabular}{c|cc|ccc|ccc|ccc}
				\toprule			
				&\multicolumn{2}{c}{Training} & \multicolumn{3}{c}{In-sample}&\multicolumn{3}{c}{Out-of-sample 1} &\multicolumn{3}{c}{Out-of-sample 2} \\
				\hline
				&Time (s/Iter) & $\#$ Iters & $F(\theta^\ast)$& \makecell{Avg rel. \\error} & \makecell{Max rel. \\error} & $F(\theta^\ast)$& \makecell{Avg rel. \\error} & \makecell{Max rel. \\error} & $F(\theta^\ast)$& \makecell{Avg rel. \\error} & \makecell{Max rel. \\error} \\
				\hline 
				\makecell{Deep BSDE \\ (Converged)} &  18.0359& 13 & 12.3758& 0.0136& 0.0866& 20.0577& 0.0202& 0.1227& 12.1314& 0.0157& 0.1072\\				
				\hline
				\makecell{Monte Carlo \\ (Time constrained)}  & 40.4468& 6 & 59.7336& 0.0297& 0.1483& 63.3884 & 0.0357& 0.1917& 71.3379& 0.0390 & 0.1768\\				
				\hline
				\makecell{Monte Carlo \\ (Converged)} &  40.4468 &  200& 5.6158& 0.0091& 0.0650& 9.7465& 0.0141& 0.0950& 6.6558 & 0.0118& 0.0798\\				
				\bottomrule
			\end{tabular}
		\end{adjustbox}
		\caption{Calibration error for the deep BSDE scheme and the Monte Carlo method. The number of samples generated per iteration for Monte Carlo is $2^14$. The learning rate for both schemes is set to be $0.0004$, Patience = 1.}
		\label{tab: calibration hist 1}	
	\end{table}
	
	\begin{table}[htp]
		\centering
		\begin{tabular}{c cccc}
			\toprule
			& $\xi_0^\ast$ & $H^\ast$ & $\rho^\ast$ & $\eta^\ast$ \\
			\hline
			Deep BSDE & 0.0448& 0.0452& -0.6052 & 2.4948\\
			Monte Carlo & 0.0464 & 0.0535 & -0.6470&  2.5066\\ 
			\bottomrule
		\end{tabular}
		\caption{The set of best calibrated model parameters out of 10 runs of deep BSDE scheme and pure Monte Carlo.}	
		\label{tab: calibrated model param 1}	
	\end{table}
	
	\subsubsection{Model parameters as time-varying functions} 
	To enhance the model's learning capability, we replace some model parameters with 5-segment piecewise constant functions. Table \ref{tab: calibration hist 2} compares the performance of four distinct parameterization strategies and  the reported results represent the mean of 10 independent runs. Case 1, which parameterizes only the initial forward variance curve $\xi_0(t)$, fails to yield superior in-sample or out-of-sample performance compared to the baseline model (where all parameters are scalars), as shown in Row 1 of Table \ref{tab: calibration hist 1}. In contrast, parameterizing $H(t)$ incurs higher computational costs but outperforms all other models on both in-sample and out-of-sample datasets. This suggests that $H$ plays a more critical role than other parameters in identifying option prices. The calibrated model parameters are plotted in Figure 3, where "Case 0" denotes the scenario in which all parameters are treated as constants.
	\begin{table}[htp]
		\centering
		\begin{adjustbox}{max width=\textwidth}
			\begin{tabular}{c|c|ccc|ccc|ccc}
				\toprule			
				&&\multicolumn{3}{c}{In-sample}&\multicolumn{3}{c}{Out-of-sample 1} &\multicolumn{3}{c}{Out-of-sample 2} \\
				\hline
				&Time (s/Iter) & $F(\theta^\ast)$& \makecell{Avg rel. \\error} & \makecell{Max rel. \\error} & $F(\theta^\ast)$& \makecell{Avg rel. \\error} & \makecell{Max rel. \\error} & $F(\theta^\ast)$& \makecell{Avg rel. \\error} & \makecell{Max rel. \\error} \\
				\hline 
				Case 1: $\xi_0(t)$& 18.9727& 12.3960& 0.0136& 0.0869& 20.3964&0.0208&0.1320& 14.8027& 0.0174&0.1165 \\				
				\hline
				Case 2: $H(t)$ & 26.4407 & 12.2178 & 0.0134 & 0.0807 & 19.8499 & 0.0199 & 0.1176 & 12.0461 & 0.0154 & 0.1025 \\
				\hline				
				Case 3: $\xi_0(t), H(t)$ & 29.1549& 12.2122& 0.0134& 0.0808& 20.2548& 0.0205& 0.1267& 14.7771& 0.0171&  0.1116\\
				\hline
				Case 4: $\xi_0(t), \eta(t)$ & 22.1040& 12.3540& 0.0136& 0.0878&20.6784& 0.0209& 0.1329& 15.1855& 0.0176& 0.1175\\				
				\bottomrule
			\end{tabular}
		\end{adjustbox}
		\caption{Calibration errors of time-varying functions. The learning rate for is set to be $0.0004$, Patience = 1.}
		\label{tab: calibration hist 2}	
	\end{table}
	
	\begin{figure}[htp]
		\includegraphics[width=0.95\linewidth]{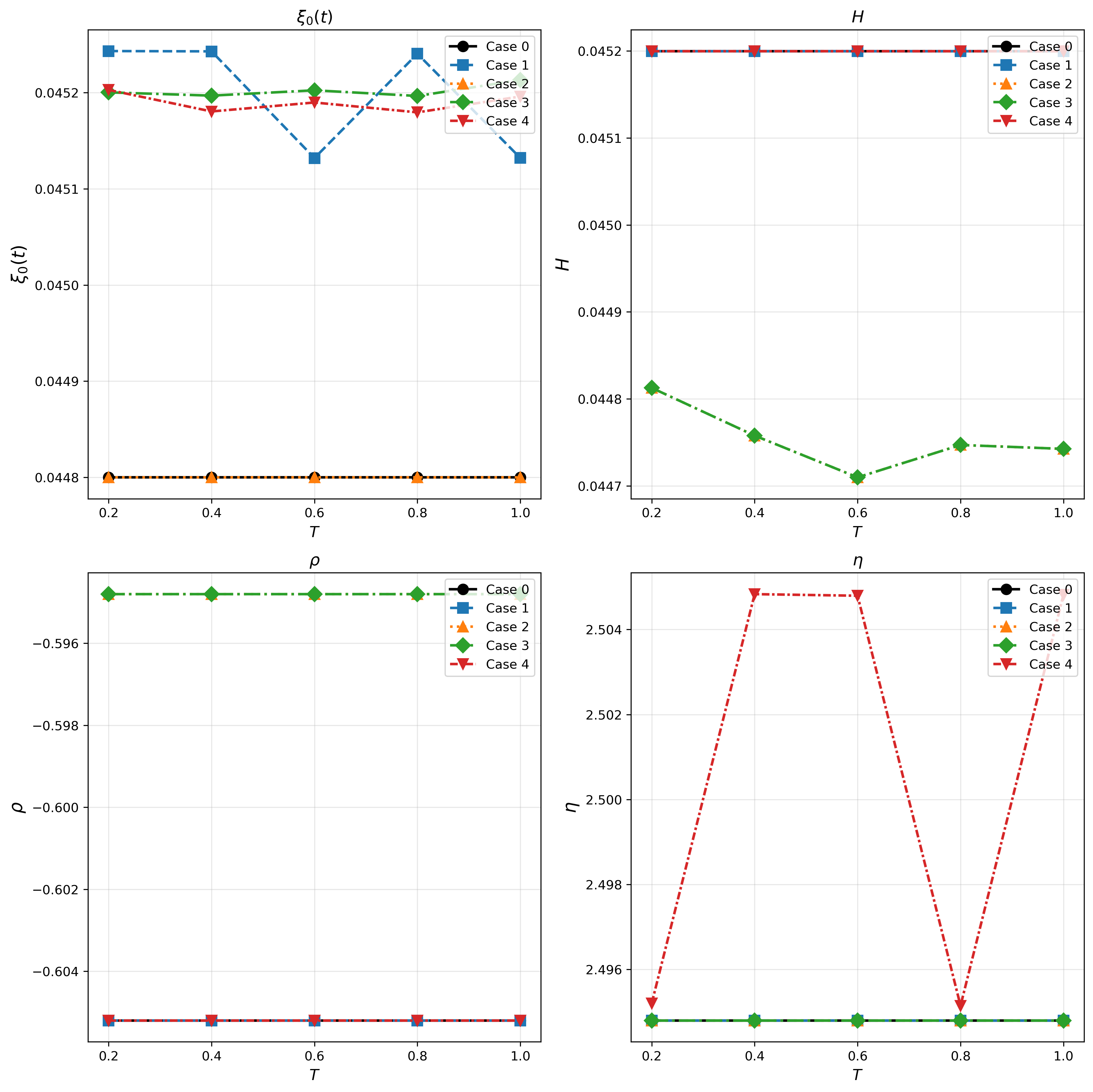}
		\centering
		\caption{Calibrated model parameters under different parameterizations.}
		\label{fig: calibrated parameters}
	\end{figure}
	
	\section{Conclusion}\label{sec:conclusion}
	In this study, we proposed a novel unsupervised learning framework for calibrating rough volatility models that eliminates the need for labeled training data. The methodology utilizes historical market data as the initial condition to simulate the corresponding Backward Stochastic Differential Equation (BSDE) forward in time, employing the terminal discrepancy as the loss function. We utilize neural networks to approximate the unknown solutions of the BSDE with increasing input dimensions and treat model parameters as trainable variables to be optimized concurrently. We established rigorous upper bounds on the discrepancy between historical data and the prices implied by the learned model parameters, expressed as a function of the loss. Furthermore, we demonstrated that this loss can be minimized arbitrarily, contingent upon the market fitting capacity of rough volatility models and the universal approximation properties of neural networks.
	
	Future research will focus on developing a dimension-stationary numerical scheme for the Backward Stochastic Partial Differential Equation (BSPDE). The current approach links the BSPDE solution to a BSDE and employs neural networks with increasing input dimensions to approximate the BSDE solutions, following the implementation described in \cite{bayer2022pricing}. Consequently, the maximal dimension scales linearly with the reciprocal of the time step size. This scaling poses a challenge for traditional numerical methods and substantially increases the computational complexity of the neural network architecture. Additionally, we intend to generalize the rough volatility model to incorporate time-dependent parameters. Numerical experiments provided evidence that a time-varying Hurst index enhances the model's fitting capability. Further research is required to extend the results presented in \cite{bayer2022pricing} and the convergence analysis in Section \ref{sec: convergence} to this time-dependent parameter setting.
	\paragraph{Acknowledgment.} The authors are very grateful to the anonymous referees for their line-by-line and valuable comments which have greatly helped the authors to improve the work.  
	
	\appendix
	\section{mSOE scheme for time-dependent rBergomi model} \label{append: mSOE scheme for time-dependent rBergomi}
	We extend the modified sum-of-exponentials (mSOE) scheme, originally proposed in \cite{teng2025efficient}, to the rBergomi model with time-dependent parameters, where the model parameters are functions of time $t$. The dynamics is given by
	\begin{align}
		\de S_t &= r(t)S_t\de t + S_t\sqrt{V_t}\br{\rho(t)\de W_t + \sqrt{1 -\rho(t)^2}\de W_t^\perp}, &&S_0 = s_0, \nonumber \\ 
		V_t &= \xi_0(t)\exp\br{\eta(t)\sqrt{2H(t)}\int_0^t\br{t - s}^{H(t) - \frac{1}{2}}\de W_s - \frac{\eta(t)^2}{2}t^{2H(t)}}, &&V_0 = v_0, \label{eq: variance of time-dependent rBergomi}
	\end{align}
	defined on a risk-neutral probability space $(\Omega, \calF, \br{\calF}_{t \in [0, T]}, \bbQ)$. Let the log price process be defined as $X_t = -\int_0^t r(s)\de s + \ln S_t$. To generate the necessary samples of $(X_t, V_t, W_t, W_t^\perp)$, we consider an equidistant temporal grid $0 = t_0 < t_1 < \cdots < t_n = T$ with time step size $h := T/n$ and $t_i = ih$. A primary challenge in sampling $(V_{t_i})_{i=1}^n$ arises from the stochastic integral within the exponential function, denoted as
	\begin{align*}
		I(t_i) := \sqrt{2H(t_i)}\int_0^{t_i}\br{t_i - s}^{H(t_i) - \frac{1}{2}}\de W_s.
	\end{align*}
	The fundamental principle of the mSOE scheme is to treat the kernel $K_i(x):= x^{H(t_i) - \frac{1}{2}}$ analytically near its singularity at $x = 0$ and approximate it using a sum of exponentials elsewhere. Specifically, let $\hat{K}_i$ be the kernel approximation defined as	 
	\begin{align*}
		\hat{K}_i(x) :=\left\{
		\begin{array}{ll}
			x^{H(t_i) - \frac{1}{2}}    & x \in [t_0, t_1), \\	
			\sum_{j = 1}^{N_{\exp}}\omega_j(t_i)\e^{-\lambda_j(t_i)x}  & x \in [t_1, t_n],   
		\end{array} \right.	
	\end{align*}
	where $\br{\lambda_j(x_i)}_{j = 1}^{\Nexp}$ is a set of nodes and $\br{\omega_j(x_i)}_{j = 1}^{\Nexp}$ are the corresponding weights, both of which implicitly depend on $t_i$. For simplicity, we assume the number of summation terms $\Nexp$ is identical for each kernel $K_i$. Based on this approximation, $I(t_i)$ is estimated as the sum of a local part and a historical part:
	\begin{align*}
		\bar{I}(t_i):&= \underbrace{\sqrt{2H(t_i)}\int_{t_{i-1}}^{t_i} (t_i - s)^{H(t_i) - \frac{1}{2}}\de W_s}_{I_{\calN}(t_i)} + \underbrace{\sqrt{2H(t_i)}\sum_{j = 1}^{\Nexp}\omega_j(t_i)\int_{0}^{t_{i-1}}\e^{-\lambda_j(t_i)(t_i - s)}\de W_s}_{\bar{I}_\calF(t_i)}. 
	\end{align*}
	We define $\bar{I}_\calF^j$ to be the $j$th historical factor for $j = 1, \cdots, \Nexp$. The weighted sum of these factors constitutes the historical component, such that
	\begin{align*}
		\bar{I}_\calF^j(t_i) := \int_0^{t_{i-1}} \e^{-\lambda_j(t_i)(t_i - s)}\de W_s,\quad \bar{I}_\calF(t_i) = \sqrt{2H(t_i)}\sum_{j=1}^{N_{\text{exp}}}\omega_j(t_i)\bar{I}_\calF^j(t_i).
	\end{align*} 
	Through direct calculation, we obtain
	\begin{align*}
		\bar{I}^j_\calF(t_i) = \e^{-\lambda_j(t_i)h}\br{\int_0^{t_{i-2}}\e^{-\lambda_j(t_i)\br{t_{i-1} - s}}\de W_s + \int_{t_{i-2}}^{t_{i-1}} \e^{-\lambda_j(t_i)\br{t_{i-1} - s}}\de W_s}.
	\end{align*}
	The first stochastic integral within the brackets is Gaussian and perfectly correlated with $\bar{I}_\calF^j(t_{i-1}) = \int_0^{t_{i-2}}\e^{-\lambda_j(t_{i-1})\br{t_{i-1} - s}}\de W_s$ so it is a scalar multiple of $\bar{I}_\calF^j(t_{i-1})$. Thus, we can write
	\begin{align*}
		\bar{I}^j_\calF(t_i) = \e^{-\lambda_j(t_i)\tau}\br{\calV^j(t_i)\bar{I}_\calF^j(t_{i-1}) + \int_{t_{i-2}}^{t_{i-1}} \e^{-\lambda_j(t_i)\br{t_{i-1} - s}}\de W_s}, 
	\end{align*}
	where the scaling factor is given by	
	\begin{align*}
		\calV^j\br{t_i} :&= \frac{\bbE\fa{\br{\int_0^{t_{i-2}}\e^{-\lambda_j(t_i)\br{t_{i-1} - s}}\de W_s}^2}^{1/2}}{\bbE\fa{\br{\bar{I}_\calF^j(t_{i-1})}^2}^{1/2}} \\
		&= \br{\frac{\lambda_j(t_{i-1})}{\lambda_j(t_i)}\frac{\e^{-2\lambda_j(t_i)t_{i-1}} - \e^{-2\lambda_j(t_i)\tau}}{\e^{-2\lambda_j(t_{i-1})t_{i-1}} - \e^{-2\lambda_j(t_{i-1})\tau}}}^{1/2}, 
	\end{align*}
	for $i = 3, \cdots, n, j = 1, \cdots, \Nexp$. We set $\calV^j\br{t_2} = 1$ since $\bar{I}_\calF^j(t_1) = 0$ by definition. This yields the following recursive formula for each historical factor
	\begin{equation*}
		\bar{I}_\calF^j(t_i) =\left\{
		\begin{aligned}
			&0 && i = 1, \\
			&\e^{-\lambda_j(t_i)\tau}\br{\calV^j(t_i)\bar{I}_\calF^j(t_{i-1}) + \int_{t_{i-2}}^{t_{i-1}} \e^{-\lambda_j(t_i)\br{t_{i-1} - s}}\de W_s} &&  i \geq 2.
		\end{aligned} 
		\right.
	\end{equation*}
	Given this recurrence, we must simulate a centered $\br{\Nexp + 2}$-dimensional Gaussian random vector at time $t_i$ for $i = 1, \cdots, n-1$	
	\begin{align*}
		\Xi_i := \br{\Delta W_{t_i}, \int_{t_{i-1}}^{t_i}\e^{-\lambda_1(t_{i+1})\br{t_i - s}}\de W_s, \cdots, \int_{t_{i-1}}^{t_i}\e^{-\lambda_{\Nexp}(t_{i+1})\br{t_i - s}}\de W_s, I_\calN(t_i)}, 
	\end{align*}
	where $\Delta W_{t_i} = W_{t_i} - W_{t_{i-1}}$ is included  allow for the joint simulation of $X_t$ and $V_t$. At the terminal time $t_n = T$, simulation is required only for the 2-dimensional Gaussian vector $\Xi_n := \br{\Delta W_{t_n}, I_\calN(t_n)}$. The entries of the covariance matrix $\Sigma^i$ of the Gaussian vector $\Xi_i$ are defined as follows:
	\begin{align*}
		& \Sigma^i_{1, 1} = h, \quad \Sigma^i_{1, \ell} = \Sigma^i_{\ell, 1} = \frac{1 - \e^{-\lambda_{\ell-1}(t_{i+1})h }}{\lambda_{\ell-1}(t_{i+1})}, \\
		& \Sigma^i_{1, \Nexp+2} = \Sigma^i_{\Nexp+2, 1} = \frac{\sqrt{2H(t_i)}h^{H(t_i) + 1/2}}{H(t_i) + 1/2}, \\	
		& \Sigma^i_{k, \ell} = \Sigma^{i}_{\ell, k} = \frac{1 - \e^{-\br{\lambda_{\ell-1}(t_{i+1}) + \lambda_{k-1}(t_{i+1})}h}}{\lambda_{\ell-1}(t_{i+1}) + \lambda_{k-1}(t_{i+1})}, \\	
		& \Sigma^i_{\Nexp+2, \ell} = \Sigma^i_{\ell, \Nexp+2} = \frac{\sqrt{2H(t_i)}}{\lambda_{\ell-1}^{H + 1/2}(t_{i+1})}\gamma(H(t_i) + 1/2, \lambda_{\ell-1}(t_{i+1})h),\\
		& \Sigma^i_{\Nexp + 2, \Nexp+2} = h^{2H(t_i)},
	\end{align*}
	for $k, \ell = 2, \cdots, \Nexp+1$, where $\gamma(\cdot, \cdot)$ is the lower incomplete gamma function. Since $\Sigma^i$ varies w.r.t. $i$, we must implement Cholesky decomposition at each time step. Furthermore, as the term $t^{2H(t)}$ in \ref{eq: variance of time-dependent rBergomi} corresponds to the second moment of $I(t_i)$, we replace it by $\bbE[\bar{I}(t_i)^2]$:
	\begin{align*}
		\bbE[\bar{I}^2(t_i)] = h^{2H(t_i)} + 2H(t_i)\sum_{k,\ell=1}^{\Nexp}\frac{\omega_k(t_i)\omega_\ell(t_i)}{\lambda_k(t_i) + \lambda_\ell(t_i)}(\e^{-(\lambda_k(t_i) + \lambda_\ell(t_i))h} - \e^{-(\lambda_k(t_i) + \lambda_\ell(t_i))t_i}).
	\end{align*}
	Finally, the mSOE scheme is summarized as follows for $i = 1, \cdots, n$	\begin{enumerate}
		\item Compute nodes $(\lambda_j(t_i))_{j=1}^{\Nexp}$ and weights $(\omega_j(t_i))_{j=1}^{\Nexp}$ for each kernel $K_i$. 
		\item Implement Cholesky decomposition for $\Sigma^{i}$ and sample the Gaussian vector $\Xi_{i}$ together with $W_{t_i}^\perp$.
		\item Update historical components:
		\begin{align*}
			\bar{I}_{\calF}^j(t_{i+1}) = \e^{-\lambda_j(t_{i+1})h}\br{\calV^j(t_{i+1})\bar{I}^j_\calF(t_i) + \Xi_i^{(j+1)}},\quad j = 1, \cdots, \Nexp,
		\end{align*}
		where $\Xi_i^{(j+1)}$ denotes the $(j+1)$th entry of $\Xi_i$.
		\item Compute the variance process:
		\begin{align*}
			\bar{V}_{t_i} = \xi_0(t_i)\exp(\eta(t_i)\br{\Xi_i^{(\Nexp + 2)} + \sqrt{2H(t_i)}\sum_{j=1}^{\Nexp}\bar{I}_\calF^j(t_i)} - \frac{\eta(t_i)^2}{2}\bbE\fa{\bar{I}^2(t_i)}),\quad \bar{V_0} = v_0.		
		\end{align*}
		\item Update the log asset price process by Euler-Maruyama:
		\begin{align*}
			\bar{X}_{t_i} = \bar{X}_{t_{i-1}} - \frac{1}{2}\bar{V}_{t_{i-1}}h + \sqrt{\bar{V}_{t_{i-1}}}\br{\rho(t_{i-1})\Delta W_{t_i} + \sqrt{1 - \rho(t_{i-1})^2}\Delta W_{t_i}^\perp},\quad \bar{X}_0 = \log s_0.
		\end{align*}		
	\end{enumerate}
	\paragraph{Computational Complexity:} The mSOE scheme entails an offline cost of $\calO\br{n\Nexp^3}$ for the Cholesky decompositions and an online computational cost of $\calO(nN_{\text{exp}})$, with a storage requirement of $\calO(N_{\text{exp}})$ per path.
	
	\bibliographystyle{abbrv}
	\bibliography{myref}
	
\end{document}